\documentclass[journal,12pt,onecolumn,draftclsnofoot,]{IEEEtran}

\usepackage{blindtext}
\usepackage{wrapfig}
\usepackage{amsmath,amsfonts,amssymb}
\usepackage{algorithm,algorithmic}
\usepackage{bbm,balance,array}
\usepackage[caption=false,font=normalsize,labelfont=sf,textfont=sf]{subfig}
\usepackage{color}
\usepackage{comment,cite,dsfont}
\usepackage{graphicx}
\graphicspath{{plots/}}
\usepackage{mathtools,}
\usepackage{psfrag}
\usepackage{epstopdf}
\usepackage{epsfig}
\usepackage{multicol}
\usepackage{multirow}
\usepackage{tabularx}
\usepackage[table]{xcolor}
\usepackage{stfloats}
\usepackage{url}
\usepackage{verbatim}

\ifCLASSINFOpdf

\else

\fi

\DeclareMathAlphabet      {\mathbfit}{OML}{cmm}{b}{it}

\newcommand{\Mod}[1]{\ (\mathrm{mod}\ #1)}
\newtheorem{theorem}{Theorem}
\newtheorem{lemma}{Lemma}

\newtheorem{corollary}[theorem]{Corollary}

\newtheorem{remark}{Remark}
\newtheorem{example}{\textbf{Example}}
\newenvironment{proof}{}{\hfill\rule{2mm}{2mm}}
\newcommand{\remove}[1]{}
\begin{document}
%
\title{Structured Index Coding Problem \\and Multi-access Coded Caching}

\author{Kota Srinivas Reddy and Nikhil Karamchandani\\ 
	Department of Electrical Engineering, \\
	 Indian Institute of Technology, Bombay \\
	Email: ksreddy@ee.iitb.ac.in, nikhilk@ee.iitb.ac.in
	\vspace{-0.3in}
}

\newcommand\myeq{\stackrel{\mathclap{\normalfont\mbox{(\small a)}}}{=}}

\maketitle

\begin{abstract}
	Index coding and coded caching are two active research topics in information theory with strong ties to each other. Motivated by the multi-access coded caching problem, we study a new class of structured index coding problems (ICPs) which are formed by the union of several symmetric ICPs. We derive upper and lower bounds on the optimal server transmission rate for this class of ICPs and demonstrate that they differ by at most a factor of two. Finally, we apply these results to the multi-access coded caching problem to derive better bounds than the state of the art.
\end{abstract}
\IEEEpeerreviewmaketitle
\section{Introduction}\label{sec:icp_introduction}
{\let\thefootnote\relax\footnote{Preliminary version of this work appeared in  \cite{reddy2020structured}. This work was supported in part by  a SERB grant on ``Content Caching and Delivery over
			Wireless Networks" and seed grant from IIT Bombay.		
}}
Index coding is a fundamental problem in network information theory \cite{birk1998informed} which consists of a central server with a collection of messages, communicating with a set of users over a broadcast channel. Each user has prior knowledge of a subset of the messages, referred to as side-information, and is interested in recovering another subset of the messages. The goal of the index coding problem (ICP)  is to minimize the server transmission size while satisfying all the user requests. While the problem in general remains open, several bounds are known based on min rank \cite{bar2011index}, local chromatic number \cite{shanmugam2013local}, local partial clique cover \cite{agarwal2016local} and maximum acyclic induced subgraph (MAIS)\cite{arbabjolfaei2013capacity}. There has been work on characterizing the optimal transmission rate for specific classes of ICPs such as the uniprior ICP \cite{ong2016uniprior}, single unicast index coding problem (SUICP) with symmetric neighboring and consecutive (SUICP-SNC) side-information \cite{maleki2014index,vaddi2019minrank}, SUICP with symmetric consecutive interference (SUICP-SCI) \cite{vaddi2018broadcast} and the ICP with interlinked cycle structure \cite{thappa2017interlinked}.
  An ICP is said to be symmetric \cite{maleki2014index}, if the relative positions of side-information and interference users are the same for all the users. In this work, we focus on a new class of structured ICPs which are formed by the union of several symmetric ICPs and are motivated by the multi-access coded caching problem (MACC)  \cite{hachem2017codedmulti}.

ICP also has close relations to several other  active research topics, including network coding  \cite{effros2015equivalence}, distributed storage \cite{mazumdar2015storage}, distributed computing \cite{li2017fundamental}, and coded caching \cite{maddah2016coding}. The coded caching problem was introduced in \cite{maddah2014fundamental}, which pioneered an information-theoretic view of the classical caching problem, and helps in reducing the network congestion by offloading the data traffic from peak to off-peak hours. The setup in \cite{maddah2014fundamental} (Ali-Niesen setup) consists of a central server with $N$ files, each of size 1 unit, communicating over an error-free broadcast link with $K$ users each with a cache of size $M$ units. There are two phases in the system. The first phase is the placement phase, in which the content related to the files is stored in the caches. The second phase is the delivery phase, in which each user requests a file from the central server, and these requests are served by the central server with the help of the caches. The aim is to minimize the central server transmission while ensuring that each user can recover its requested file using the server's transmission message and the cache content user has access to. \cite{maddah2014fundamental} proposed a (uncoded{\footnote{In the uncoded placement policy, we are allowed to split the files into parts and store the individual file parts, but coding across the file parts is not allowed while storing in the caches.}}) placement and (coded) delivery policy, which is order-wise optimal up to a constant factor of 12 with respect to the information-theoretic lower bound. Following this, the coded caching problem has attracted a lot of attention in the information theory community. Several works focus on improving the gap between the achievable rates and the lower bounds, see for example \cite{maddah2016coding} for an extensive survey. In fact, \cite{yu2018exact,wan2016optimality} showed that the policy proposed in \cite{maddah2014fundamental} is exactly optimal under the restriction of uncoded placement. In particular, \cite{wan2016optimality} exploited the connection between the coded caching and the index coding problems to prove the exact optimality.

A generalization of the Ali-Niesen setup is the multi-access coded caching (MACC) problem introduced in \cite{hachem2017codedmulti}, where each user is connected to $L$ consecutive caches with a cyclic wrap-around, see Section \ref{sec:macc_setting} for the detailed description of the problem setup. \cite{hachem2017codedmulti} proposed a (uncoded) placement and (coded) delivery policy and showed that the multiplicative gap between their achievable rate and the information-theoretic lower bound is bounded by $c\cdot L$, for some constant $c$. Following this, \cite{reddy2020rate,sasi2020improved,serbetci2019multi,cheng2020novel} focus on improving the achievable rate and deriving the order-optimal bound. \cite{reddy2020rate,sasi2020improved,serbetci2019multi,cheng2020novel} primarily use the uncoded placement policies and coded delivery policies.

In the second part of this paper, we establish a connection between the MACC problem to the class of structured symmetric ICPs studied in the first part, and use the results derived to establish new bounds on the optimal rate-memory trade-off of the caching problem.  
 If the placement policy and the request pattern are fixed, then a coded caching problem is equivalent to solving an ICP. In our work, we use the same uncoded placement policy proposed in \cite{reddy2020rate} and then use our ICP results in the delivery policy to derive the new bounds on the optimal rate-memory trade-off. 
 We analytically show that our achievable rate can be better than the achievable rates in \cite{hachem2017codedmulti,reddy2020rate,cheng2020novel}. Furthermore, we also show via examples and numerical evaluations that there exist system parameter regimes where our achievable rate can be better than those proposed in \cite{sasi2020improved,serbetci2019multi}.
 
 Recently, the MACC problem has also been studied in \cite{sasi2021multi,mahesh2020coded,katyal2021multi,muralidhar2021multi}. In particular,  \cite{sasi2021multi,mahesh2020coded} focus on the MACC with linear sub-packetization. The works \cite{katyal2021multi,muralidhar2021multi} focus on a MACC problem with a slightly altered topology based on cross-resolvable designs. \cite{katyal2021multi,muralidhar2021multi} consider a new metric rate-per-user and showed that the schemes from cross-resolvable designs perform better under this new metric than the schemes in \cite{maddah2014fundamental,serbetci2019multi}.

Our contributions in this paper are summarized below: 
\begin{enumerate}
	\item We define a new class of structured  ICPs which are formed by the union of several symmetric ICPs. We derive upper and lower bounds on the optimal transmission rate for this class and show that they differ by at  most a multiplicative factor of 2.
	\item We derive the exact optimal transmission rate of some special symmetric ICPs.
	\item We apply our ICP results to derive a new bound on the achievable rate of the MACC problem and show that our achievable rate can be better than the achievable rates in \cite{hachem2017codedmulti,reddy2020rate,sasi2020improved,serbetci2019multi,cheng2020novel} using examples and numerical evaluations.
	\item For the MACC problem with access degree $L\ge K/2$, we show that our proposed achievable rate is at most a factor 5/4 away from the optimal rate under the restriction of uncoded placement. This improves the previously best known multiplicative gap of 2 \cite{reddy2020rate}.
\end{enumerate}

The rest of the paper is organized as follows: Sections \ref{sec:not} and \ref{sec:icp_setting} define some useful notations and describe the ICP setting studied in this paper. Sections \ref{sec:ICP_preliaries} and \ref{ICP_results} include some ICP preliminaries and present the main results respectively. Section \ref{sec:macc_setting} describes the   MACC setting  and presents the improved upper bound on the optimal rate. Numerical results are given in Section \ref{sec:numerical} and the summary  of our work and future possible extensions are given in \ref{sec:conclusions}. All the proofs are relegated  to Appendix.

\section{Notations}\label{sec:not}

\begin{itemize}
	\item $[n]=\{1,2,3,...,n\}$
	\item For some given{\footnote{In our manuscript, $K$ is the number of users in the index coding problem / multi-access coded caching problem.}}  integer $K$ such that $K\geq m, n$, \begin{align*}
	[m:n]=\begin{cases}
	\{m,m+1,...,n\} & \text{ if } m\leq n \\
	\{m,m+1,...,K,1,2,...,n\} & \text{ if } m>n 
	\end{cases}
	\end{align*}
	\item \begin{align*}
	<m>_n=\begin{cases}
	m\Mod{n} & \text{ if } m\Mod{n}\neq 0 \\
	n & \text{ if } m\Mod{n}=0 
	\end{cases}
	\end{align*}
	\item $|S|-$ size of file / subfile / set $S$ 
	\item {${F}_{i,\mathcal{S}}$} denotes parts of File $i$ exclusively available to users with index in set $\mathcal{S}$	
	\item $len(\mathbf{v})$ denotes the number of components  of the vector $\mathbf{v}$
	\item $\mathbf{\widehat{a}}$ denotes the maximum value amongst the components of the vector $\mathbf{a}$
\end{itemize}

\section{Setting} \label{sec:icp_setting}
 An index coding problem (ICP)  consists of $N$ files $\mathcal{X}=\{x_1,x_2,...,x_N\}$, each of size 1 unit ($=F$ bits\footnote{We assume that the file size $F$ is sufficiently large.}) at a central server and a set of $K$ users $\mathcal{U}=\{U_1,U_2,...,U_K\}$. Each user $U_k=(\mathcal{W}_k,\mathcal{K}_k)$  wants a subset of files $\mathcal{W}_k\subseteq \mathcal{X}$, and we  call it the \textit{want-set} of $U_k$. Each user $U_k$ knows another subset of files $\mathcal{K}_k\subseteq\mathcal{W}_k^c $, and we  call it the \textit{known-set} or side-information of $U_k$.  The remaining subset of files $(\mathcal{W}_k\cup \mathcal{K}_k)^c$, i.e., the files which are neither requested nor available, are called the \textit{interference-set} of User $k$. The central server transmits a message such that each user $k$ recovers its  \textit{want-set} $\mathcal{W}_k$ using its \textit{known-set} $\mathcal{K}_k$ and the central server's transmitted message. The goal is to minimize the server transmission size while satisfying all the user requests. The minimum transmission size required to satisfy all the user demands in an ICP, over all possible coding strategies, is called the optimal transmission rate $R^*$ of the ICP. 
 An ICP is said to be a unicast index coding problem (UICP) if none of the users want the same file, i.e., $\mathcal{W}_{k_1}\cap\mathcal{W}_{k_2}=\phi$, $\forall k_1\neq k_2$ and an ICP is said to be a single unicast index coding problem (SUICP) if every user wants only one file and it is a UICP, i.e., $|\mathcal{W}_{k}|=1$, $\forall k$ and $\mathcal{W}_{k_1}\cap\mathcal{W}_{k_2}=\phi$, $\forall k_1\neq k_2$.

We refer to an ICP as an $(a_i,a_{i-1},...,a_1)_L-$ICP for some  non-negative integers $a_1,a_2,...,a_i$ and a natural number $L$ such that 
\begin{align}\label{eq:sumcondn}
(i-1)L+\sum_{j=1}^{i}a_j=K-1,
\end{align} if $\forall k\in[K]$, User $k$
\begin{itemize}
	\item  wants $x_k$, i.e., $\mathcal{W}_k=\{x_k\}$ and
	\item knows \begin{align*}
	\mathcal{K}_k=\{x_b:b=<k+\sum_{j=1}^{v}a_{i+1-j}+(v-1)L+r>_K,
	 v\in[i-1], r\in[L]\},
	\end{align*}
	 i.e., the side-information $\mathcal{K}_k$ is available in $i-1$ chunks,  each of which is a collection of $L$ consecutive elements and the separation/gap between the consecutive chunks is determined by $a_i, a_{i-1}, \ldots, a_1$. More explicitly, the side information $\mathcal{K}_k$ of User $k$ is given by
\end{itemize}
\begin{align*}
\{x_{<k+a_i+1>_K},x_{<k+a_i+2>_K},...,x_{<k+a_i+L>_K}\}
\cup\hspace{0.3in}\\
\{x_{<k+a_i+L+a_{i-1}+1>_K},..., x_{<k+a_i+a_{i-1}+2L>_K}\}\cup\hspace{0.3in}\\
\vdots  \hspace{1.5in}\\
\cup \{x_{<k+\sum_{j=0}^{i-2}a_{i-j}+(i-2)L+1>_K},..., x_{<k+\sum_{j=0}^{i-2}a_{i-j}+(i-1)L>_K}\}.
\end{align*} 
%
Note that the $(a_i,a_{i-1},...,a_1)_L-$ICP is an SUICP and $|\mathcal{K}_k|=(i-1)L$. In an $(a_i,a_{i-1},...,a_1)_L-$ ICP, if we arrange the files $x_1, x_2,...,x_K$ circularly in clock-wise direction, then User $k$ wants File $x_k$ and if we go on from $x_k$ in clock-wise direction, first we see $a_i$ interference files for User $k$, then $L$ side-information files, then $a_{i-1}$ interference  files,  then $L$ side-information files and so on. It ends with $a_1$ interference files for User $k$. See Figure \ref{fig:211_2} for an illustration.
\begin{figure}[ht]
	\begin{center}
		\includegraphics[scale=0.33]{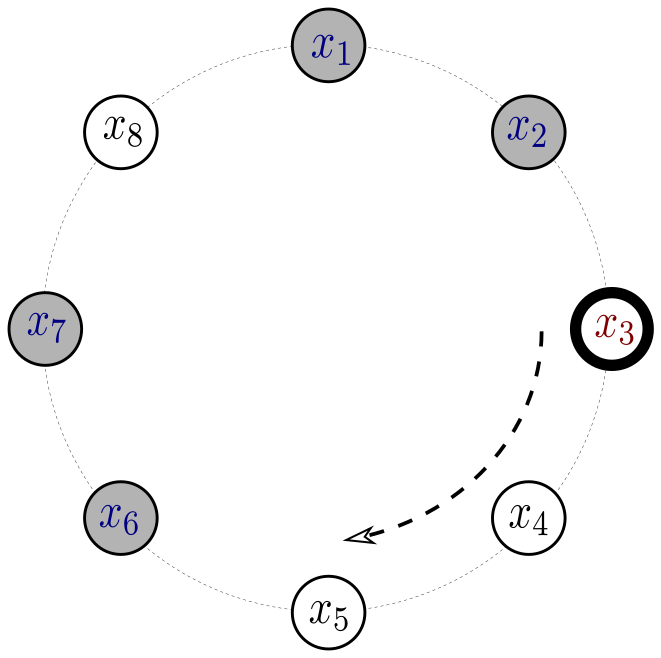}
		\vspace{-0.1in}
		\caption{\sl An illustration of the $(2,1,0)_2-$ICP.  It has $i = 3$, $L = 2$, $a_3 = 2$, $a_2 = 1$, $a_1 = 0$ and the number of users $K=(i-1)L + \sum_{j=1}^{i} a_j + 1 =  8$. Here, we highlight User 3's requested file $x_3$ with a thick circle and red color font. Among the other files, the shaded circles with blue color fonts represent the side-information files of user 3 which are available at User 3, and  the empty circles denote the interference files of User 3 which are not available at User 3. If we go on from $x_3$ in clock-wise direction, we first see $a_3=2$ interference files $x_4, x_5$, then $L=2$ side-information files $x_6, x_7$, then $a_2=1$ interference file $x_8$, then $L=2$ side-information files $x_1, x_2$ and lastly $a_1=0$ interference files of User 3. Observe that the side-information files are available in groups/chunks of size $L=2$.\label{fig:211_2}}
	\end{center}
	\vspace{-0.2in}
\end{figure}

\begin{example}
	Figure \ref{fig:211_2} shows the $(2,1,0)_2-$ICP. It contains eight users such that\\ $U_1=(\{x_{1}\}, \{x_{4},x_{5},x_{7},x_{8}\})$, $U_2=(\{x_{2}\}, \{x_{5},x_{6},x_{8},x_{1}\})$, $U_3=(\{x_{3}\}, \{x_{6},x_{7},x_{1},x_{2}\})$,\\ $U_4=(\{x_{4}\}, \{x_{7},x_{8},x_2,x_3\})$, $U_5=(\{x_{5}\}, \{x_{8},x_{1},x_3,x_4\})$, $U_6=(\{x_{6}\}, \{x_{1},x_2,x_4,x_5\})$,\\ $U_7=(\{x_{7}\},$ $\{x_{2},x_{3},x_5,x_6\})$, $U_8=(\{x_{8}\},  \{x_{3},x_{4},$ $x_6,x_7\})$.\\ Recall that in an ICP, each user $U_k$ with the required files $\mathcal{W}_k$ and the side-information files $\mathcal{K}_k$ is denoted as $U_k=(\mathcal{W}_k,\mathcal{K}_k)$.
\end{example}


In this paper, we focus on a union of multiple SUICPs which is defined as follows. Let there be $K$ users such that User $k$ wants $\mathcal{W}_k^j$ and knows $\mathcal{K}_k^j$ in the $j^{th}$ instance for some $j\in[i]$, $i\in\mathbb{N}$, then the union of these $i$ instances is an ICP with 
\begin{itemize}
	\item want set $\mathcal{W}_k=\mathcal{W}^1_k\cup \mathcal{W}^2_k\cup...\cup \mathcal{W}^i_k$ and
	\item known set $\mathcal{K}_k=\mathcal{K}^1_k\cup \mathcal{K}^2_k\cup...\cup \mathcal{K}^i_k$.
\end{itemize}
We assume that there are no common files involved across the instances.


\remove{In general, we can add any number of SUICPs. }In this paper, we focus on a UICP formed by the union of $i\in\mathbb{N}$ special SUICPs. We use the notation $x_{k,j}$ to represent the $j^{th}$ file requested by User $k$.

 We refer to an ICP as the $\overline{(a_i,a_{i-1},...,a_1)}_L-$ICP, if it is a union of the $(a_i,a_{i-1},...,a_1)_L-$ICP, $(a_1,a_i,...,a_2)_L-$ICP, and so on till $(a_{i-1},a_{i-2},...,a_{1},a_i)_L-$ICP, i.e., an $\overline{(a_i,a_{i-1},...,a_1)}_L-$ICP is formed by the union of the $(a_i,a_{i-1},...,a_1)_L-$ICP and its $i-1$ clockwise rotations. 
 In an $\overline{(a_i,a_{i-1},...,a_1)}_L-$ICP,  $a_j$ is a non-negative integer $\forall j\in[i]$, $L$ is a natural number and they satisfy \eqref{eq:sumcondn}. In the $\overline{(a_i,a_{i-1},...,a_1)}_L-$ICP, $\forall k\in [K]$
 \begin{itemize}
 	\item  want set $\mathcal{W}_k=\{x_{k,1}, x_{k,2}, ..., x_{k,i}\}$, and
 	\item  known set \begin{align}\label{eq:side-info}
 		\mathcal{K}_k=\{x_{b,t}:t\in[i], v\in[i-1], r\in[L],b= <k+\sum_{j=1}^{v}a_{<i+t-j>_i}+(v-1)L+r>_K
 		\},
 	\end{align}
 	i.e., $$\mathcal{K}_k=\mathcal{K}^1_k\cup \mathcal{K}^2_k\cup \hdots \cup \mathcal{K}^i_k,$$
 	where 
 	$\mathcal{K}_k^p$  is the side-information structure of $(a_{p-1},a_{p-2}...,a_{1},a_i,...,a_p)_L-$ICP, which denotes the $(p-1)^{th}$ clock-wise rotation of ${(a_i,a_{i-1},...,a_1)}_L-$ICP, and is given by
 \end{itemize} 
 	\begin{align}\label{eq:side-info2}
 		\mathcal{K}^p_k=\{x_{b,p}:b=<k+\sum_{j=1}^{v}a_{<i+p-j>_i}+(v-1)L+r>_K, v\in[i-1], r\in[L]\}.
 	\end{align}
 
 
 \begin{example}\label{example210}
 	A $\overline{(2,1,0)}_{2}-$ICP is the union of $(2,1,0)_{2}-$ICP, $(0,2,1)_{2}-$ICP and $(1,0,2)_{2}-$ICP. Note that  $\overline{(2,1,0)}_{2}-$ICP has $i = 3$, $L = 2$, $a_3 = 2$, $a_2 = 1$, $a_1 = 0$ and the number of users $K=(i-1)\times L + \sum_{j=1}^{i} a_j + 1=  8$, each user wants $i=3$ files and has $i(i-1)L=12$ files as side-information. In the $\overline{(2,1,0)}_2-$ICP, $\forall k\in [8]$,
 	\begin{itemize}
 		\item  want set $\mathcal{W}_k=\{x_{k,1}, x_{k,2}, x_{k,3}\}$, and
 		\item  known set \begin{align}\label{eq:210kk}
 			\mathcal{K}_k=\{x_{b,t}:t\in[3], v\in[2], r\in[2],b=
 			<k+\sum_{j=1}^{v}a_{<i+t-j>_i}+(v-1)L+r>_K
 			\},
 		\end{align} where $a_1=0, a_2=1, a_3=2$.
 	\end{itemize}
 \end{example}

Note that in an $\overline{(a_i,a_{i-1},...,a_1)}_L-$ICP, each user wants $i$ files and has $i(i-1)L$ files as side-information. Since $\overline{(a_i,a_{i-1},...,a_1)}_L-$ICP, $\overline{(a_1,a_i,...,a_2)}_L-$ICP, \ldots, $\overline{(a_{i-1},a_{i-2},...,a_{1},a_i)}_L-$ICP all are  equivalent, we assume without loss of generality (WLOG) 
$a_i\geq a_j$ $\forall j\in[i]$.

An $\overline{(a_i,a_{i-1},...,a_1)}_L-$ICP can be represented by a $K\times i$ table such that the $k^{th}$ row and $j^{th}$ column contains the $j^{th}$ file requested by User $k$.  We represent a general $\overline{(a_i,a_{i-1},...,a_1)}_L-$ICP in Table \ref{Tab:genicp} and the $\overline{(2,1,0)}_2-$ICP in Table \ref{Tab:211_2}.  We denote the $k^{th}$ row and $j^{th}$ column's entry as Node $(k,j)$. In these tables, $k^{th}$ row represents User $k$'s requested files. In Table \ref{Tab:211_2}, we  highlight User 3's requested files with red color, bold fonts and its side-information files with blue color fonts and shaded nodes. A Node $(m_1,n_1)$ is said to be side-information node of any Node $(m_2,n_2)$, if User $m_2$ has File $x_{m_1,n_1}$ in its side-information. For example, the shaded nodes in Table \ref{Tab:211_2} are side-information nodes of {Node  $(3,1)$. Recall that an $\overline{(a_i,a_{i-1},...,a_1)}_L-$ICP is formed by the union of the $(a_i,a_{i-1},...,a_1)_L-$ICP and its $i-1$ clockwise rotations. According to our tabular representation in Table \ref{Tab:genicp}, the first column corresponds to the $(a_i,a_{i-1}, ..., a_1)_L-$ICP, the second column corresponds to the $(a_1, a_i, a_{i-1}, ..., a_2)_L-$ICP and so on till the $i^{th}$ column which corresponds to the $(a_{i-1}, a_{i-2},..., a_1,a_i)_L-$ICP. In Column 1 of Table \ref{Tab:211_2}, User 3's requested file is $x_{3,1}$ and it is followed by the two interference files $x_{4,1}$ and $x_{5,1}$, then two side-information files $x_{6,1}$ and $x_{7,1}$, then one interference file $x_{8,1}$, then two side-information files $x_{1,1}$ and $x_{2,1}$ and it ends with the 0 interference files. From symmetry, the same structure holds in the first column for all other users as well and thus we can conclude that the first column corresponds to $(2,1,0)_2-$ICP. Similarly, we can verify that the second column corresponds to $(0,2,1)_2-$ICP and the third column corresponds to $(1,0,2)_2-$ICP.

\begin{table}[t]
		\vspace{0.8cm}
	\centering
	\begin{tabular}{| c | c | c | c |}
		\hline 
		$x_{1,1}$ &  $x_{1,2}$ & ... & $x_{1,i}$\\
		\hline 
		$x_{2,1}$ &  $x_{2,2}$ & ... & $x_{2,i}$\\
		\hline 
		\vdots & \vdots &\vdots &\vdots \\
		\hline
		$x_{K,1}$ &  $x_{K,2}$ & ... & $x_{K,i}$\\
		\hline 
	\end{tabular}
	\vspace{0.2cm}
	\caption{\sl The tabular representation of the  $\overline{(a_i,a_{i-1},...,a_1)}_L-$ICP. $x_{k,j}$ represents the $j^{th}$ file requested by User $k$. In the table, the first column corresponds to the $(a_i,a_{i-1}, ..., a_1)_L-$ICP, the second column corresponds to the $(a_1, a_i, a_{i-1}, ..., a_2)_L-$ICP and so on till the  $i^{th}$ column which corresponds to the $(a_{i-1}, a_{i-2},..., a_1,a_i)_L-$ICP. } \label{Tab:genicp}
	\end{table}	
\begin{table}[t]
\hspace{0.01\linewidth}
	\centering
	\begin{tabular}{| c | c | c |}
		\hline 
		\cellcolor{gray}\color{blue}${x_{1,1}}$ &  \cellcolor{gray}\color{blue}${x_{1,2}}$ &  ${x_{1,3}}$\\
		\hline 
		\cellcolor{gray}\color{blue}$x_{2,1}$ &  $x_{2,2}$ & $x_{2,3}$\\
		\hline
		\color{red}$\mathbf{x_{3,1}}$ & \color{red} $\mathbf{x_{3,2}}$ & \color{red}$\mathbf{x_{3,3}}$\\
		\hline
			$x_{4,1}$ &  \cellcolor{gray}\color{blue}$x_{4,2}$ & $x_{4,3}$\\
		\hline 
			$x_{5,1}$ &  \cellcolor{gray}\color{blue}$x_{5,2}$ & \cellcolor{gray}\color{blue}$x_{5,3}$\\
		\hline
		\cellcolor{gray}\color{blue}$x_{6,1}$ &  $x_{6,2}$ & \cellcolor{gray}\color{blue}$x_{6,3}$\\
		\hline
		\cellcolor{gray}\color{blue}	$x_{7,1}$ & $x_{7,2}$ & \cellcolor{gray}\color{blue}$x_{7,3}$\\
		\hline
			$x_{8,1}$ &  \cellcolor{gray}\color{blue}$x_{8,2}$ & \cellcolor{gray}\color{blue}$x_{8,3}$\\
		\hline
	\end{tabular}
	\vspace{0.2cm}
	\caption{\sl This table corresponds to  the $\overline{(2,1,0)}_2-$ICP. Here, we highlight the requested files of User 3 with the red color fonts and the side-information files of User 3 with the shaded cells and the blue color fonts. The remaining files are the interference files for User 3.  Here, Columns 1 to 3 correspond to the $(2,1,0)-$ICP, $(0,2,1)-$ICP and $(1,0,2)-$ICP respectively. } \label{Tab:211_2}
\vspace{-0.75cm}
\end{table}

Our aim is to find the optimal broadcast rate of an $\overline{(a_i,a_{i-1},...,a_1)}_L-$ICP.

\section{Preliminaries}\label{sec:ICP_preliaries}
There are several upper and lower bounds \cite{arbabjolfaei2013capacity,shanmugam2013local,agarwal2016local,bar2011index} available for an SUICP. We use local chromatic number and maximum acyclic induced sub-graph based bounds in our proofs and they are described below.

In an SUICP, User $j$ is said to be interfering with User $k$, if User $j$ is requesting  a file from $\mathcal{K}_k^c$, i.e., User $j$'s requested file is not available at User $k$. The closed anti-outneighborhood of User $k$  is defined as the set containing User $k$ itself and all its interfering users and is denoted by $\mathcal{N}^+(k)$. A proper coloring scheme assigns a color to each user such that no user shares its color with any of its interfering users.  For an SUICP with $m$ users, let $c:[m]\rightarrow[n]$ for some positive integer $n\leq m$ be a proper coloring scheme with $n$ colors and let $c(\mathcal{N}^+(k))$ denote the set of different colors assigned to the closed anti-outneighborhood of User $k$ under coloring scheme $c$. Then the local chromatic number of the ICP ($\mathcal{X}_l$) is defined as
$$\mathcal{X}_l=\min_c\max_{k\in[m]}|c(\mathcal{N}^+(k))|.$$
In words, the local chromatic number of an ICP is defined as
the maximum number of different colors that appear in any user's closed anti-outneighborhood, minimized over all proper coloring schemes.

\begin{lemma}\label{lemma:icpupperbound}
	For a given SUICP, let $\mathcal{X}_l$ denotes its local chromatic number and $R^*$ denotes its optimal broadcast rate, then $R^*\leq \mathcal{X}_l.$
\end{lemma}

Lemma \ref{lemma:icpupperbound} gives an upper bound on the optimal broadcast rate of an SUICP and it follows from \cite[Theorem~1]{shanmugam2013local}. \cite{shanmugam2013local} also discusses the explicit construction of server transmission schemes based on MDS codes which achieve this bound.  

An SUICP with $m$ users can equivalently be represented by a side information graph $\mathcal{G}$ with $m$ nodes such that each node represents a unique user and there exists an edge from Node $k$ to Node $j$ if  User $j$'s requested file belongs to  \textit{Known-set} $\mathcal{K}_k$ of User $k$.

\begin{lemma}\label{lemma:icplowerbound}
	For a given SUICP with side-information graph $\mathcal{G}$, the optimal broadcast rate $R^*$ is greater than or equal to $MAIS(\mathcal{G})$, where $MAIS(\mathcal{G})$ is the size of the maximum acyclic induced sub-graph of $\mathcal{G}$.
\end{lemma} 
Lemma \ref{lemma:icplowerbound} gives a lower bound on the optimal broadcast rate of an SUICP and it follows from \cite[Theorem~3]{bar2011index}.
\section{ICP main results}\label{ICP_results}
In this section, we discuss our main results related to the $\overline{(a_i,a_{i-1},...,a_1)}_L-$ICP. Let the variable $R_u$ be defined as 
	\begin{align}\label{eqn:rub}
		R_u=\min\{2(K-(i-1)L)+i-2-a_i,K\}.
	\end{align}
	The following theorem gives an upper bound on the optimal rate of the $\overline{(a_i,a_{i-1},...,a_1)}_L-$ICP.
\begin{theorem}\label{thm:icpub}
	Consider an $\overline{(a_i,a_{i-1},...,a_1)}_L-$ICP with $K$ users and $R_u$ be defined as in \eqref{eqn:rub}. Let $R^*$ be the  optimal transmission rate of the ICP, then
	$$R^*\leq R_u=\min\{2(K-(i-1)L)+i-2-a_i,K\}.$$
\end{theorem}
As mentioned in Section \ref{sec:ICP_preliaries}, the local chromatic number gives an upper bound on the SUICP. So, we convert a UICP with $K$ users into an SUICP with $iK$ virtual users such that each user in the UICP maps into $i$ virtual users in the SUICP, and each virtual user requests a distinct file of the original user's requested files. The side-information at the virtual user is the same as its corresponding real user. We find an upper bound on the local chromatic number by assigning a proper coloring scheme to the virtual users. This upper bound works as an upper bound for the broadcast rate of the original problem. One naive coloring scheme is to assign a unique color to each virtual user. Note that for this coloring scheme, the number of colors assigned to the closed anti-outneighborhood of any virtual user is $iK - i(i-1)L$. This is because the total number of users is $iK$ and the size of the side-information set for any user is $i(i-1)L$.  
In this paper, we propose a better coloring scheme, which gives the upper bound shown in Theorem \ref{thm:icpub}. The details are given in Appendix \ref{sec:icpub}. Here, we discuss  the coloring scheme and upper bound for the $\overline{(2,1,0)}_2-$ICP mentioned in Example \ref{example210}. Note that for the $\overline{(2,1,0)}_2-$ICP, the bound corresponding to the naive coloring scheme is $iK-i(i-1)L=12$ whereas our proposed scheme gives a bound of $2(K-(i-1)L)+i-2-a_i=7$. 


\addtocounter{example}{-1}
\begin{example} (continued)
In  the $\overline{(2,1,0)}_2-$ICP,  there are $K=8$ users, each one requesting $i=3$ files, and has $i(i-1)L=12$ files as side-information.
The $\overline{(2,1,0)}_2-$ICP is shown in the tabular form in Table \ref{Tab:211_2}. Note that the $\overline{(2,1,0)}_2-$ICP is a UICP with eight users, and each user requesting three distinct files. We convert this UICP into an SUICP with 24 virtual users, each one requesting  a distinct file. In particular, each user in the UICP is  mapped to 3 virtual users in the SUICP, such that each virtual user requests a distinct file of the original user's requested files.  The side-information at the virtual user is the same as 
its corresponding original user. In Table \ref{Tab:211_2}, recall that we use Node $(k,j)$ to represent the $j^{th}$ file requested by User $k$. Now, we call Node $(k,j)$ as $j^{th}$ virtual user of User $k$. The side-information of Node $(k,j)$ is the same as side-information of User $k$. It is easy to see that the optimal broadcast rate in the two ICPs (the virtual SUICP and the original UICP) are equal. We give an upper bound for the SUICP, and it also works as an upper bound for the original UICP. 

Recall from Section \ref{sec:ICP_preliaries} that for an SUICP, a proper coloring scheme assigns a color to each user such that no user shares its color with any of its interfering users.
We take $K=8$ colors and assign Color $c$ to nodes $(<c+\sum_{j=1}^{v-1}a_{j}+(v-1)L>_K,v)$ $\forall v\in[3]$, i.e., to Node ($c,1$) in Column 1, to Node ($<c+a_1+L>_K,2$) in Column 2, to Node ($<c+a_1+a_2+2L>_K,3$) in Column 3. The coloring scheme  is shown in the tabular form in Table \ref{Tab:211_2color}. Note that every color occurs in a column exactly once.
\begin{table}[h]
	\begin{minipage}{.240\linewidth}
	\centering
	\begin{tabular}{|c|c|c|}
		\hline 
	\color{red}\textbf{1} & \color{red}\textbf{7} & \color{red}\textbf{4}\\
		\hline 
		2 &\cellcolor{gray}\color{blue}8&  5\\
		\hline 
		3 & \cellcolor{gray}\color{blue} 1 &\cellcolor{gray}\color{blue}6\\
		\hline 
\cellcolor{gray}	\color{blue}	4 & 2 &\cellcolor{gray}\color{blue}7\\
		\hline
\cellcolor{gray}\color{blue}5 &  3 &\cellcolor{gray}\color{blue}8\\
		\hline
		6 &  \cellcolor{gray}\color{blue}4 &\cellcolor{gray}\color{blue}1\\
		\hline
\cellcolor{gray}\color{blue}7 & \cellcolor{gray}\color{blue}5 &2\\
		\hline
\cellcolor{gray}\color{blue}8 & 6 &  3\\
		\hline
	\end{tabular}\\
\vspace{0.2cm}
(i)
	\end{minipage}%
\begin{minipage}{.240\linewidth}
	\centering
	\begin{tabular}{| c | c | c |}
		\hline 
		\cellcolor{gray}\color{blue}1 &7 &4\\
		\hline 
		\color{red}$\mathbf{2}$&\color{red}\textbf{8}&\color{red}\textbf{5}\\
		\hline 
		3 & \cellcolor{gray}\color{blue}1&6\\
		\hline 
		4 & \cellcolor{gray}\color{blue}2&\cellcolor{gray}\color{blue}7\\
		\hline
		\cellcolor{gray}\color{blue}5&  3 &\cellcolor{gray}\color{blue}8\\
		\hline
		\cellcolor{gray}\color{blue}6&  4 &\cellcolor{gray}\color{blue}1\\
		\hline
		\color{blue}	7 &\cellcolor{gray}\color{blue}5&\cellcolor{gray}\color{blue}2\\
		\hline
		\cellcolor{gray}\color{blue}8 &\cellcolor{gray}\color{blue}6&  3 \\
		\hline
	\end{tabular}\\
\vspace{0.2cm}
(ii)
\end{minipage}
\begin{minipage}{.240\linewidth}
	\centering
	\begin{tabular}{| c | c | c |}
		\hline 
		\cellcolor{gray}\color{blue}1 &  \cellcolor{gray}\color{blue}7 &  4\\
		\hline 
		\cellcolor{gray}\color{blue}2 &  8 &  5\\
		\hline 
	\color{red}	\textbf{3} & \color{red} \textbf{1} & \color{red} \textbf{6}\\
		\hline 
		4 & \cellcolor{gray}\color{blue}2 & 7\\
		\hline
	5 & \cellcolor{gray}\color{blue}3 &  \cellcolor{gray}\color{blue}8\\
		\hline
		\cellcolor{gray}\color{blue}6 &  4 & \cellcolor{gray}\color{blue}1\\
		\hline
		\cellcolor{gray}\color{blue}7 & 5 & \cellcolor{gray}\color{blue}2\\
		\hline
		8 & \cellcolor{gray}\color{blue}6 &  \cellcolor{gray}\color{blue}3\\
		\hline
	\end{tabular}\\
	\vspace{0.2cm}
	(iii)
\end{minipage}
\begin{minipage}{.240\linewidth}
	\centering
	\begin{tabular}{| c | c | c |}
		\hline 
		1 &\cellcolor{gray}\color{blue}7 &\cellcolor{gray}\color{blue}4\\
		\hline 
		\cellcolor{gray}\color{blue}2 &\cellcolor{gray}\color{blue}8 &  5\\
		\hline 
		\cellcolor{gray}\color{blue}3 &1 & \color{blue} 6\\
		\hline 
			\color{red}\textbf{4} &	\color{red}\textbf{2} & 	\color{red}\textbf{7}\\
		\hline
		5 & \cellcolor{gray}\color{blue}3&  8\\
		\hline
		6 & \cellcolor{gray}\color{blue}4&\cellcolor{gray}\color{blue}1\\
		\hline
		\cellcolor{gray}\color{blue}7 & 5&\cellcolor{gray}\color{blue}2\\
		\hline
		\cellcolor{gray}\color{blue}8 &6&\cellcolor{gray}\color{blue}3\\
		\hline
	\end{tabular}\\
	\vspace{0.2cm}
	(iv)
\end{minipage}
%
%
\begin{minipage}{.240\linewidth}
	\begin{tabular}{| c | c | c |}
	\hline 
	\cellcolor{gray}\color{blue}1 & 7 &\cellcolor{gray}\color{blue}4\\
	\hline 
	2 &\cellcolor{gray}\color{blue}8 & \cellcolor{gray}\color{blue}5\\
	\hline 
	\cellcolor{gray}\color{blue}3 &\cellcolor{gray}\color{blue}1&6\\
	\hline 
	\cellcolor{gray}\color{blue}4 &2&7\\
	\hline
	\color{red}	\textbf{5} &\color{red} \textbf{3} &\color{red} \textbf{8}\\
	\hline
	6 & \cellcolor{gray}\color{blue}4 &1\\
	\hline
	7 & \cellcolor{gray}\color{blue}5 &\cellcolor{gray}\color{blue}2\\
	\hline
	\cellcolor{gray}\color{blue}8 & 6 &\cellcolor{gray}\color{blue}3\\
	\hline
\end{tabular}\\

\hspace{0.75cm}(v)
\end{minipage}%
\begin{minipage}{.240\linewidth}
\centering
\begin{tabular}{| c | c | c |}
	\hline 
	\cellcolor{gray}\color{blue}1 & 7 &\cellcolor{gray}\color{blue}4\\
	\hline 
	\cellcolor{gray}\color{blue}2 &  8 &\cellcolor{gray}\color{blue}5\\
	\hline 
	3 &\cellcolor{gray}\color{blue}1 &\cellcolor{gray}\color{blue}6\\
	\hline 
	\cellcolor{gray}\color{blue}4 & \cellcolor{gray}\color{blue}2 &7\\
	\hline
\cellcolor{gray}\color{blue}5 &  3 &  8\\
	\hline
\color{red}	\textbf{6} &\color{red}	\textbf{4}&\color{red}	\textbf{1}\\
	\hline
7 & \cellcolor{gray}\color{blue}5 & 2\\
	\hline
8 & \cellcolor{gray}\color{blue}6 & \cellcolor{gray}\color{blue}3\\
	\hline
\end{tabular}\\
\vspace{0.2cm}
(vi)
\end{minipage}
\begin{minipage}{.240\linewidth}
\centering
\begin{tabular}{| c | c | c |}
	\hline 
1 & \cellcolor{gray}\color{blue}7 &\cellcolor{gray}\color{blue}4\\
	\hline 
	\cellcolor{gray}\color{blue}2 &  8 &\cellcolor{gray}\color{blue}5\\
	\hline 
	\cellcolor{gray}\color{blue}3 & 1 &\cellcolor{gray}\color{blue}6\\
	\hline 
	4 & \cellcolor{gray}\color{blue}2 & \cellcolor{gray}\color{blue}7\\
	\hline
\cellcolor{gray}\color{blue}5 &\cellcolor{gray}\color{blue}3&  8\\
	\hline
	\cellcolor{gray}\color{blue}6 &  4 & 1\\
	\hline
	 \color{red}\textbf{7} &  \color{red} \textbf{5} & \color{red} \textbf{2}\\
	\hline
		8 & \cellcolor{gray}\color{blue}6 &  3\\
	\hline
\end{tabular}\\
\vspace{0.2cm}
(vii)
\end{minipage}
\begin{minipage}{.240\linewidth}
\centering
\begin{tabular}{| c | c | c |}
	\hline 
	1 & \color{red} \cellcolor{gray}\color{blue}7 &  4\\
	\hline 
	2 &  \cellcolor{gray}\color{blue}8 &  \cellcolor{gray}\color{blue}5\\
	\hline 
	\cellcolor{gray}\color{blue}3 & 1 & \cellcolor{gray}\color{blue}6\\
	\hline 
	\cellcolor{gray}\color{blue}4 & 2 & \cellcolor{gray}\color{blue}7\\
	\hline
	5 &  \cellcolor{gray}\color{blue}3 &  \cellcolor{gray}\color{blue}8\\
	\hline
	\cellcolor{gray}\color{blue}6 &  \cellcolor{gray}\color{blue}4 & 1\\
	\hline
	\cellcolor{gray}\color{blue}7 & 5& 2\\
	\hline
	\color{red}	\textbf{8} &\color{red}	\textbf{6} &\color{red}	\textbf{3}\\
	\hline 
\end{tabular}\\
\vspace{0.2cm}
(viii)
\end{minipage}
	\vspace{0.2cm}
	\caption{\sl Coloring scheme for the $\overline{(2,1,0)}_2-$ICP. The number in the cell indicates the color assigned to the cell. In Table \ref{Tab:211_2color}($m$), we discuss the Row $m$ cells, and their colors are highlighted with the red color bold fonts, their side-information cells are highlighted with the shaded cells and blue color fonts, normal cells represent their interference nodes. Observe that any red-colored bold font color always occurs at the shaded cell in the other columns. Hence it is a proper coloring scheme. The number of distinct colors in the non-shaded cells is 7 (local chromatic number) in every table. } \label{Tab:211_2color}
	\vspace{-0.75cm}
\end{table}

For Node $(m,1)$, we assign Color $m$. We also assign Color $m$ in Column 2 to the node $(<m+a_1+L>_K, 2)$, which is equal to  node $(<m+2>_8, 2)$ and we can verify that it  belongs to $\mathcal{K}_m$ by substituting $k=m, t=2, v=1, r=2, K=8 $ in \eqref{eq:210kk}.  We assign Color $m$ in Column 3 to the node $(<m+a_1+a_2+2L>_K, 3)$, which is equal to  node $(<m+5>_8, 3)$ and  we can verify that it  belongs to $\mathcal{K}_m$ by substituting $k=m, t=3, v=2, r=2, K=8 $ in \eqref{eq:210kk}. Hence, we can conclude that the colors assigned to Column 1 nodes are  only assigned to side-information nodes in the other columns.  Similar arguments can be used to check that for any node ($m, n$), the color assigned to it is only shared with its side-information nodes. Table \ref{Tab:211_2color} illustrates it row by row. Therefore, this coloring scheme ensures that none of the nodes share its color with its interfering nodes and hence is a proper coloring scheme.

%
%

In Table \ref{Tab:211_2color} (i), for User 1, according to the definition of closed anti-outneighborhood in Section \ref{sec:ICP_preliaries}, colors in the non-shaded cells contribute to local chromatic number. They are $K-(i-1)L=4$ in Column 1 (colors 1, 2, 3 and 6), extra $a_1+1=1$  color (Color 7) are added in Column 2, and finally $a_2+1=2$  colors (colors 4 and 5) are added in Column 3. Hence, the closed anti-outneighborhood of User 1 contain $K-(i-1)L+a_1+1+a_2+1=2(K-(i-1)L)+i-2-a_i=7$ colors $\{1,2,3,4,5,6,7\}$. By symmetry, this property holds true for other users as well and it can be verified easily using Table \ref{Tab:211_2color}. In particular, we can check that User $m$ contains  $\{m,<m+1>_8,<m+2>_8,<m+3>_8,<m+4>_8,<m+5>_8,<m+6>_8\}$ colors in its closed anti-outneighborhood from Table \ref{Tab:211_2color} ($m$). Hence, the local chromatic number for the $\overline{(2,1,0)}_2-$ICP is less than or equal to $2(K-(i-1)L)+i-2-a_i=7$. 
From Lemma \ref{lemma:icpupperbound}, $R^*\leq 2(K-(i-1)L)+i-2-a_i=7$ units.
	
\end{example}	
%

Recall that the upper bound in Theorem \ref{thm:icpub} is the minimum of $K$ and $2(K-(i-1)L)+i-2-a_i$. For the $\overline{(2,1,0)}_2-$ICP discussed in Example \ref{example210}, we got the upper bound as $2(K-(i-1)L)+i-2-a_i$ $(<K)$. Now, we discuss the $\overline{(3,2,1)}_2-$ICP in Example \ref{example321}, where we get $K$ $(< 2(K-(i-1)L)+i-2-a_i)$ as upper bound.

\begin{example}\label{example321}
In  the $\overline{(3,2,1)}_2-$ICP,  there are $K=11$ users, each one requesting $i=3$ files, and has $i(i-1)L=12$ files as side-information. In the $\overline{(3,2,1)}_2-$ICP, $\forall k\in [11]$,
\begin{itemize}
	\item  want set $\mathcal{W}_k=\{x_{k,1}, x_{k,2}, x_{k,3}\}$, and
	\item  known set \begin{align*}
	\mathcal{K}_k=\{x_{b,t}:t\in[3], v\in[2], r\in[2],b=
	<k+\sum_{j=1}^{v}a_{<i+t-j>_i}+(v-1)L+r>_K
	\},
	\end{align*} where $a_1=1, a_2=2, a_3=3$.
\end{itemize}
\begin{table}[h]
	\begin{minipage}{.450\linewidth}
		\centering
		\begin{tabular}{| c | c | c |}
			\hline 
			\color{red}	$\mathbf{x_{1,1}}$ &  \color{red}$\mathbf{x_{1,2}}$ &  \color{red}$\mathbf{x_{1,3}}$\\
			\hline 
			$x_{2,1}$ &  $x_{2,2}$ & $x_{2,3}$\\
			\hline
			$x_{3,1}$ & \cellcolor{gray}\color{blue} $x_{3,2}$ & $x_{3,3}$\\
			\hline
			{	$x_{4,1}$} & \cellcolor{gray}\color{blue}{ $x_{4,2}$} & \cellcolor{gray}\color{blue}$x_{4,3}$\\
			\hline 
			\cellcolor{gray}\color{blue}	$x_{5,1}$ &  $x_{5,2}$ & \cellcolor{gray}\color{blue}$x_{5,3}$\\
			\hline
			\cellcolor{gray}\color{blue}$x_{6,1}$ &  $x_{6,2}$ & $x_{6,3}$\\
			\hline
			$x_{7,1}$ & $x_{7,2}$ & \cellcolor{gray}\color{blue}$x_{7,3}$\\
			\hline
			$x_{8,1}$ & \cellcolor{gray}\color{blue} $x_{8,2}$ &\cellcolor{gray}\color{blue} $x_{8,3}$\\
			\hline
			\cellcolor{gray}\color{blue}	$x_{9,1}$ &\cellcolor{gray}\color{blue}  $x_{9,2}$ &  $x_{9,3}$\\
			\hline 
			\cellcolor{gray}\color{blue}		$x_{10,1}$ &  $x_{10,2}$ &  $x_{10,3}$\\
			\hline 
			$x_{11,1}$ &  $x_{11,2}$ &  $x_{11,3}$\\
			\hline 
		\end{tabular}
		\vspace{0.2cm}
		\caption{$\overline{(3,2,1)}_2-$ICP.  } \label{Tab:321_2}
	\end{minipage}
	\begin{minipage}{.5\linewidth}
		\centering
		\begin{tabular}{|c|c|c|}
			\hline 
			\color{red}\textbf{1} & \color{red}\textbf{9} & \color{red}\textbf{5}\\
			\hline 
			2 &10&  6\\
			\hline 
			3 & \cellcolor{gray}\color{blue} 11 &7\\
			\hline 
			4 &\cellcolor{gray}\color{blue} 1 &\cellcolor{gray}\color{blue}8\\
			\hline
			\cellcolor{gray}\color{blue}5 &  2 &\cellcolor{gray}\color{blue}9\\
			\hline
			\cellcolor{gray}\color{blue}	6 &  3 &10\\
			\hline
			7 & 4 &\cellcolor{gray}\color{blue}11\\
			\hline
			8 & \cellcolor{gray}\color{blue}5 &  \cellcolor{gray}\color{blue}1\\
			\hline
			\cellcolor{gray}\color{blue}	9 &  \cellcolor{gray}\color{blue}6 &  2\\
			\hline 
			\cellcolor{gray}\color{blue}	10 &  7 &  3\\
			\hline 
			11 &  8 &  4\\
			\hline 
			
		\end{tabular}
		\vspace{0.2cm}
		\caption{Coloring scheme for Table \ref{Tab:321_2}.  } \label{Tab:321_2color}
	\end{minipage}
\vspace{-0.5cm}
\end{table}

The ICP is shown in Table \ref{Tab:321_2}. We take $K=11$ colors and assign Color $c$ to nodes $(<c+\sum_{j=1}^{v-1}a_{j}+(v-1)2>_{11},v)$ $\forall v\in[3]$.  The coloring scheme  is shown in the tabular form in Table \ref{Tab:321_2color}. In Table \ref{Tab:321_2color}, we highlight User 1's color with red-colored bold fonts and side-information nodes of User 1 with shaded cells. Note that the color assigned to User 1 is repeated in other columns at shaded cells only, i.e., User 1's nodes are not sharing its color with interference nodes. Similarly,  we can also verify for the other users. Hence,  this coloring scheme is a proper coloring scheme.

Since we are using $K = 11$ colors in our scheme, we have from Lemma \ref{lemma:icpupperbound} that $R^* \le K = 11$. Note that for the $\overline{(3,2,1)}_2-$ICP, the other upper bound  $2(K-(i-1)L)+i-2-a_i$ value is 12, whereas  Example \ref{example321} tells that the transmission rate is smaller than or equal to $K=11$. 

\end{example}

Thus far, we have discussed the achievable schemes and the upper bounds on broadcast rate for the  $\overline{(a_i,a_{i-1},...,a_1)}_L$ $-$ICP. Now, we focus on the lower  bound on the optimal broadcast rate of the $\overline{(a_i,a_{i-1},...,a_1)}_L$ $-$ICP and check the tightness of the bounds. The following theorem gives a lower bound on the optimal broadcast rate of the $\overline{(a_i,a_{i-1},...,a_1)}_L$ $-$ICP and its proof is given in Appendix \ref{sec:icplb}. 
\begin{theorem}\label{thm:icplb}
	Consider an $\overline{(a_i,a_{i-1},...,a_1)}_L-$ICP with $K$ users. Let $R^*$ be the  optimal broadcast rate of the ICP. Then,
$$R^*\geq \sum_{j=1}^{i}(a_j+1)=K-(i-1)L+i-1. $$

\end{theorem}
The following corollary compares the bounds in Theorems \ref{thm:icpub} and \ref{thm:icplb}, and shows that the rate achieved by our proposed scheme is always within a factor of two from the optimal. 
\begin{corollary}\label{crl:comp}
	Consider an $\overline{(a_i,a_{i-1},...,a_1)}_L-$ICP with $K$ users. Let $R_u$ be defined as in \eqref{eqn:rub} and $R^*$ be the  optimal transmission rate of the ICP, then we have
	$$R_u/ R^*\leq 2.$$
\end{corollary}
{\begin{proof}
{\bf Proof:}	For an $\overline{(a_i,a_{i-1},...,a_1)}_L-$ICP, recall from Theorem \ref{thm:icpub} that $R_u\leq 2(K-(i-1)L)+i-2-a_i$ and from Theorem \ref{thm:icplb} that $R^*\geq K-(i-1)L+i-1$. Combining the two inequalities, we get  $R_u/ R^*\leq 2.$ 
\end{proof}
\subsection{Exact transmission rate}\label{sec:exact}
Corollary \ref{crl:comp} establishes that for a general $\bar{\mathbf{a}}_L$-ICP, the multiplicative gap between the achievable rate of our proposed scheme and the optimal rate is at most 2.  We can also characterize the exact optimal transmission rate for some special cases and these are discussed below. From  Theorems \ref{thm:icpub} and \ref{thm:icplb}, we get the following corollary.
\begin{corollary}\label{cor:exactspecial}
	Let $R_1^*$ and $R_2^*$ be the optimal transmission rates for the $\overline{(a_i,0,0,...,0)}_L-$ICP and $\overline{(a_i,a_{i-1},...,a_1)}_1-$ICP, respectively. Then, $$R_1^*=K-(i-1)L+i-1\text{ units}, \quad \& \quad R_2^*=K \text{ units}.$$ 
\end{corollary}
We can easily prove Corollary \ref{cor:exactspecial} from Theorems \ref{thm:icpub} and \ref{thm:icplb} using simple algebraic manipulations. The details are given in Appendix \ref{sec:corollary4proof}. 
	 By improving the coloring scheme in the proof of Theorem \ref{thm:icpub}, we get the following theorem.
\begin{theorem}\label{thm:exact2}
	Consider an $\overline{(a_2,a_1)}_L-$ICP with $K$ users such that $K$ is a multiple of $a_1+a_2+2$. Let $R^*$ be the  optimal transmission rate of the ICP, then
	$$R^*=a_1+a_2+2 \text{ units}.$$
\end{theorem}
To prove Theorem \ref{thm:icpub} for a general $\overline{(a_i,a_{i-1},...,a_1)}_L-$ICP,  we assign $K$ colors such that in a column, every color is assigned to exactly one node. Unlike this, to prove Theorem \ref{thm:exact2}, we assign only $a_1+a_2+2$ colors such that in a column, every color is assigned to exactly $K/(a_1+a_2+2)$ nodes. The details are given in Appendix \ref{sec:exactproof}.
\begin{example}
	Consider the $\overline{(2,1)}_{6}-$ICP and the $\overline{(2,1)}_{11}-$ICP. For these ICPs, their respective $K$s  $10$ and $15$ are indeed multiples of $a_1+a_2+2=5$ and hence their optimal transmission rate is $a_1+a_2+2=5$.
\end{example}


\section{Multi-access Coded Caching (MACC) Problem} \label{sec:macc_setting}
The multi-access coded caching (MACC) problem was proposed in \cite{hachem2017codedmulti} and has been studied recently in  \cite{reddy2020rate,serbetci2019multi,sasi2020improved} under the additional assumption of uncoded placement. In particular, the bounds proposed in \cite{reddy2020rate} are based on establishing a mapping between the MACC problem and a class of ICPs. Along similar lines, we use the ICP results derived in the previous section to derive a new achievable rate for the MACC problem and also compare it to other achievable rates in the literature.

\begin{figure}[h]
	\begin{center}
		\includegraphics[scale=0.33]{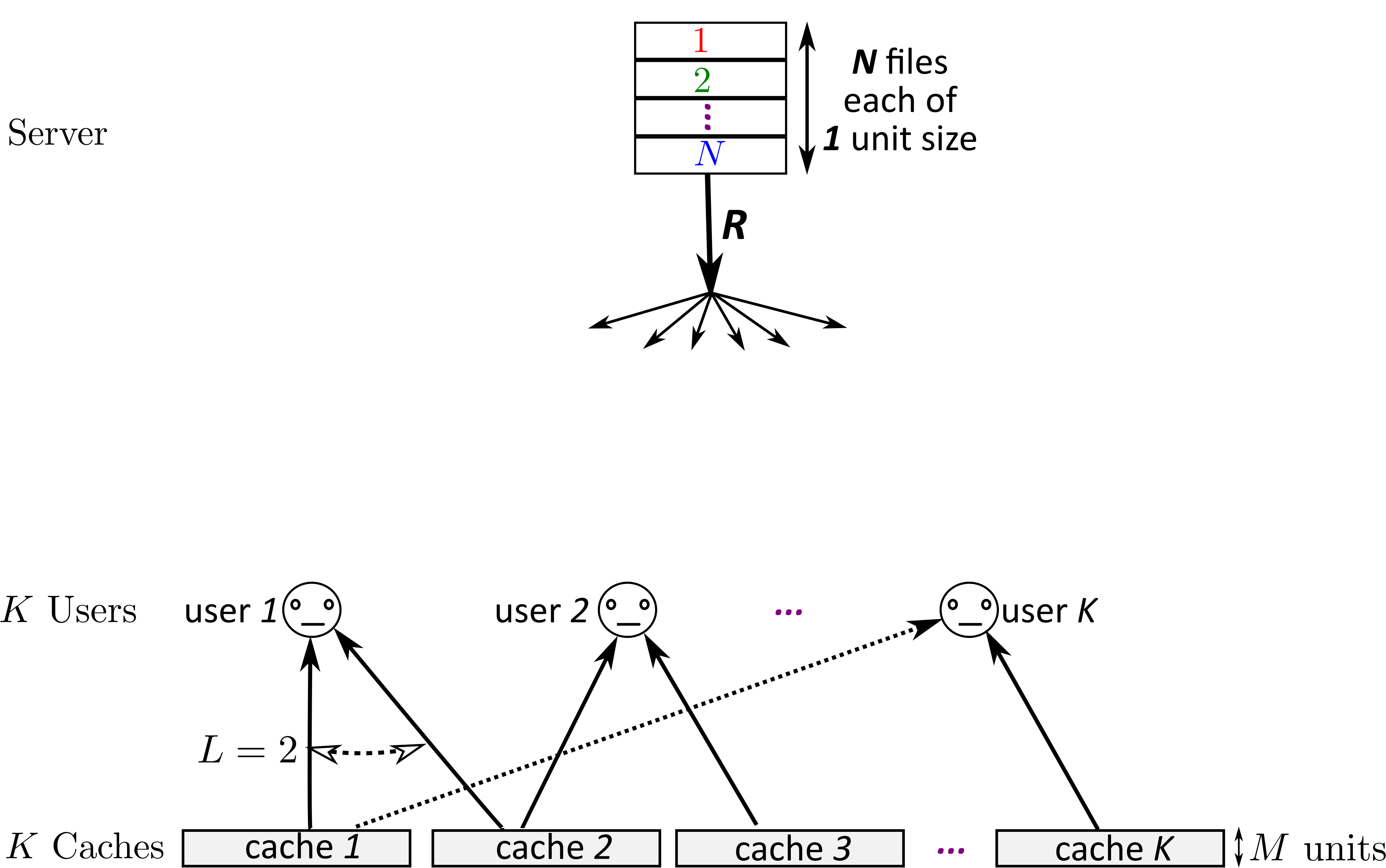}
		\caption{\sl An illustration of the ($N,K,L=2$)$-$CCDN. It consists of $N$ files, $K$ caches of $M$ units memory and $K$ users, each one connected to $L=2$ consecutive caches.   \label{fig:multiaccessps}}
	\end{center}
	\vspace{-0.2in}
\end{figure}

The MACC setup in \cite{hachem2017codedmulti} consists of a central server with $N$ files $\mathcal{F}_1, \mathcal{F}_2, ... , \mathcal{F}_N$, each of size 1 unit. There are $K$ caches, each of size $M$ units, and $K$ users, each of which has access to $L$ consecutive caches with a cyclic wrap-around, as shown in Figure \ref{fig:multiaccessps}. The system is called as the ($N,K,L$)$-$Cache aided content delivery network (CCDN) \cite{hachem2017codedmulti}. In short, it is referred as ($N,K,L$)$-$CCDN.  The system operates in two phases. The first one is the placement phase, in which we store the files according to some policy. Similar to the works in \cite{reddy2020rate,sasi2020improved,serbetci2019multi}, we restrict to uncoded placement policies wherein the caches can store the individual files and parts thereof, but no coded content is cached.  The second phase is the delivery phase, in which each user requests a file from the central server and we serve the user requests using a central server broadcast message and the caches' content. 
Our aim is to find the central server's optimal transmission rate ($R^*(M)$) for any given memory $M$ under the restriction of uncoded placement. 

 In our work, we use the same uncoded placement policy as the one proposed in \cite{reddy2020rate} and propose a better delivery policy based on the ICP results derived in Section \ref{ICP_results}. First, we consider the memory points $M$ of the form $M=wN/K$ $w\in[\lfloor K/L\rfloor]$ and derive corresponding achievable rates. The achievable rate at the intermediate points is given by memory sharing. At the extreme memory points $M=0$ units and $M=\lceil\frac{K}{L}\rceil\frac{N}{K}$ units, we can easily achieve the transmission rates $K$ units and $0$ units respectively. The details are there in \cite{reddy2020rate} and also given in Appendix \ref{sec:multiaccproof}. Now, we discuss the uncoded placement policy and delivery policy proposed for the other corner points $M=wN/K$, where $w\in[\lfloor K/L\rfloor]$. 
\subsection{Uncoded placement policy:} Let $M=wN/K$ for some $w\in[\lfloor K/L\rfloor]$ and $\hat{\mathcal{S}}$ be the collection of subsets $s$ of $[K]$, with the following constraints:
\begin{enumerate}
	\item $|s|=w$,
	\item if $w>1$, every two different elements $a_j, a_l$ of $s$ satisfy $|a_j-a_l|\geq L$ and $|K-|a_j-a_l|| \geq L$.
\end{enumerate}
Mathematically,
\begin{align*}
\hat{\mathcal{S}}=\{s=\{a_1,a_2,...,a_w\}\subseteq[K]:\forall j\neq l \text{ } |a_j-a_l|\geq L,
|K-|a_j-a_l|| \geq L \}.
\end{align*}

First, we divide each file into $|\hat{\mathcal{S}}|=\frac{K}{w}{K-wL+w-1 \choose w-1}$ equal parts and assign one subfile to each subset $s\in\hat{\mathcal{S}}$. Then, we store the subfile assigned to the set $s$, in all the $w$ caches whose index belongs to $s$. Note that according to our placement policy, each cache stores every file's ${K-wL+w-1 \choose w-1}$ parts, each of size $1/|\hat{\mathcal{S}}|$ units. The memory required to store ${K-wL+w-1 \choose w-1}$ parts of each file is $N{K-wL+w-1 \choose w-1}/|\hat{\mathcal{S}}|=wN/K=M$ units. Hence, our placement policy satisfies the memory constraint.
\subsection{Delivery policy:} The delivery phase happens after users reveal their requests. First, we form an instance of the ICP. Then, the central server transmits messages based on the solution of the ICP.

\cite{reddy2020rate} uses a naive coloring scheme, where every virtual user/node in the ICP is assigned with a different color. This naive coloring scheme \cite[Theorem 1]{reddy2020rate} gives an upper bound  on the general ($N,K,L$)$-$CCDN data transmission rate as 
\begin{align}\label{eq:oldub}
R^*(M)\leq  R_{RK}(M)=
	K\bigg(1-\frac{LM}{N}\bigg)^2. 
\end{align}

In this paper, we use the same uncoded placement policy proposed in \cite{reddy2020rate}. But in the delivery phase, instead of using the naive coloring scheme in \cite{reddy2020rate}, we use the ICP results mentioned in Section \ref{ICP_results} to get a  tighter upper bound than in \eqref{eq:oldub}. In particular, first we split our multi-access ICP into many  ICP's of the form $\overline{(a_i,a_{i-1},...,a_1)}_L-$ICP. Then, we use Theorem \ref{thm:icpub} to get an upper bound on the data transmission rate of each $\overline{(a_i,a_{i-1},...,a_1)}_L-$ICPs. Finally, the data transmission rate of multi-access ICP is upper bounded by the sum of the upper bounds of the individual ICPs. The details are given in Appendix \ref{sec:maccproof}.


Now, we discuss the upper bound given in Theorem \ref{thm:multiub2} below. A vector $\mathbf{b }= (b_{m},b_{m-1},...,b_{1})^T$ of dimension $m$ is said to be a weak $m$ compositions of $n$ \cite{stanley2011enumerative} if all the $b_i$'s are non-negative and their sum is $n$. Let $\mathcal{B}$ be the collection of all weak $w+1$ compositions  of $K-wL-1$ and $\mathbf{\widehat{b}}$ denotes the maximum component in the vector $\mathbf{b}$.  Let $R_{\text{new}}(M)$ at $M=\frac{wN}{K}$, $w\in[\lfloor K/L\rfloor]$ be  defined as
\begin{align}\label{eqn:multiub2}
	R_{\text{new}}(M)=\frac{\sum_{\mathbf{b }\in \mathcal{B}}\min\{2(K-wL)+w-1-\mathbf{\widehat{b}},K\}}{|\hat{\mathcal{S}}|(w+1)}.
\end{align}
The following theorem gives an upper bound on the transmission rate of ($N,K,L$)$-$CCDN at memory point $M=wN/K$, $w\in[\lfloor K/L\rfloor]$.

\begin{theorem}\label{thm:multiub2}
	For an ($N,K,L$)$-$CCDN, let $R^*(M)$ be the optimal transmission rate under the restriction of uncoded placement at cache size $M$ and $R_{\text{new}}(M)$ be defined as in \eqref{eqn:multiub2}. Then, at memory point $M=wN/K$, $w\in[\lfloor K/L\rfloor]$,
	$$R^*(M)\leq R_{\text{new}}(M).$$
\end{theorem}


The upper bound on the ($N,K,L$)$-$CCDN transmission rate in \cite{hachem2017codedmulti,cheng2020novel} at memory point $M=wN/K$, $w\in[\lfloor K/L\rfloor]$ is given by $R_{\text{HKD}}(M)$, where 
\begin{align}\label{eqn:colorub}
	R_{\text{HKD}}(M)=\frac{K-wL}{1+w}.
\end{align}

The following corollary compares our upper bound given in \eqref{eqn:multiub2} with the upper bounds in  \cite{hachem2017codedmulti} and \cite{reddy2020rate}. Its proof is given in Appendix \ref{sec:corollary7proof}.
\begin{corollary}\label{cor:comparison}
	Let $R_{\text{new}}(M)$, $R_{\text{HKD}}(M)$ and $R_{\text{RK}}(M)$ at $M=wN/K$, $w\in[\lfloor K/L\rfloor]$ be given by equations \eqref{eqn:multiub2}, \eqref{eqn:colorub} and \eqref{eq:oldub} respectively. Then
	$$R_{\text{new}}(M)\leq R_{\text{HKD}}(M), \quad \& \quad R_{\text{new}}(M)\leq R_{\text{RK}}(M).$$
\end{corollary}
\begin{remark}
	For an ($N,K,L=1$)$-$CCDN, $R_{\text{new}}(N/2)=\frac{K/2}{1+K/2}=\Theta(1)$ is a constant whereas $R_{\text{RK}}(N/2)=K/4=\Theta(K)$ grows linearly with $K$. On the other hand	for an ($N,K,L=K-\sqrt{K}$)$-$CCDN, $R_{\text{new}}(N/K)\leq \frac{5\sqrt{K}(\sqrt{K}+1)}{8K}=\Theta(1)$ is a constant whereas $R_{\text{HKD}}(N/K)=\sqrt{K}/2=\Theta(\sqrt{K})$ grows unbounded with $K$. Hence, our results can be order-wise better than the results in \cite{reddy2020rate,hachem2017codedmulti} in some parameter regimes.
.
\end{remark}


We also compare our upper bound given in \eqref{eqn:multiub2} with the achievable rates in \cite{sasi2020improved,serbetci2019multi} in the example below and also numerically in Section~\ref{sec:numerical}. Example \ref{ex:multiaccess} describes the key ideas of Theorem \ref{thm:multiub2} for a particular setting, the general proof is given in Appendix \ref{sec:maccproof}. 
\begin{example}\label{ex:multiaccess}
Consider the ($N, K=8, L=2$)$-$CCDN at memory point $M=2N/K$.

\textit{Placement phase:}  $w=MK/N=2$. Therefore,
$\hat{\mathcal{S}}=\{\{1,3\},\{1,4\},\{1,5\},\{1,6\},\{1,7\},\{2,4\},$ $\{2,5\},\{2,6\},\{2,7\},\{2,8\},\{3,5\},\{3,6\},\{3,7\},\{3,8\},\{4,6\},\{4,7\},\{4,8\},\{5,7\},\{5,8\},\{6,8\}\}$. Divide each file into $|\hat{\mathcal{S}}|=20$ subfiles, each of size 1/20 units, and assign one subfile to each subset. The subfile assigned to the subset $\{z_1,z_2\}$ will be stored in the caches $z_1$ and $z_2$ and will be available to the users $<z_1-1>_8,z_1,<z_2-1>_8$, and $z_2$. Therefore, we represent the subfile of File $\mathcal{F}_{j}$ stored in the caches $z_1$ and $z_2$ as $\mathcal{F}_{j,\{<z_1-1>_8,z_1,<z_2-1>_8,z_2\}}$.

\textit{Delivery phase:} Let the request pattern be $\{d_1,d_2,...,d_8\}$, i.e., User $j$ is requesting File $\mathcal{F}_{d_j}$, $\forall j\in[8]$. Out of the 20 subfiles of File $\mathcal{F}_{d_j}$, 10 subfiles stored in the caches $j$ and $j+1$ are available to User $j$. Therefore, User $j$ needs the remaining 10 subfiles, and total 80 subfiles are needed across the 8 users. 

We map the problem here to an instance of ICP with 80 virtual users/nodes such that each one requests a distinct subfile. The side-information at a virtual user is the same as the subfiles available to the real user requesting the corresponding subfile. To understand the structural properties of the ICP, we form a $8\times 10$ table such that
\begin{itemize}
	\item each cell represents a virtual user,
	\item $l^{th}$ row represents User $l$'s required subfiles,
	\item if a column's $1^{st}$ element is $\mathcal{F}_{d_1, \{z_1,<z_1+1>_8,z_2,<z_2+1>_8\}}$ then for all  $j\in[8]$, its  $j^{th}$ element is $\mathcal{F}_{d_j, \{<z_1+j-1>_8,<z_1+j>_8,<z_2+j-1>_8,<z_2+j>_8\}}$. 
\end{itemize}

Table \ref{ICPex} shows the ICP corresponding to the ($N, K=8, L=2$)$-$CCDN with $M=2N/K$. Since the same user requests all the cells in a row,   the side-information cells are the same for all the cells in a row. In particular,  all the cells which contain $l$ in the subscript are available at user $l$ and hence are side-information cells to the cells in Row $l$. In each column, the subscripts are circularly shifting by one from Row $j$ to Row $j+1$, and the cardinality of subscripts is $4$. So, the number of side-information  cells for Row $l$ cells in each column are $4$. The side-information structure for Row 1 is shown in Table \ref{siex}. 

\begin{table*}[t]
	\resizebox{\linewidth}{!}{
			\centering
\begin{tabular}{|c|c|c|c|c|c|c|c|c|c|}
	\hline
	 Column 1 & Column 2 & Column 3 & Column 4 & Column 5 & Column 6 & Column 7 & Column 8 & Column 9 & Column 10\\
	\hline
	 $\mathcal{F}_{d_1, \{2,3,4,5\}}$ & $\mathcal{F}_{d_1, \{3,4,5,6\}}$ & $\mathcal{F}_{d_1, \{4,5,6,7\}}$ & $\mathcal{F}_{d_1, \{5,6,7,8\}}$ &  $\mathcal{F}_{d_1, \{2,3,5,6\}}$ & $\mathcal{F}_{d_1, \{3,4,6,7\}}$ & $\mathcal{F}_{d_1, \{4,5,7,8\}}$ & $\mathcal{F}_{d_1, \{2,3,6,7\}}$ & $\mathcal{F}_{d_1, \{3,4,7,8\}}$ & $\mathcal{F}_{d_1, \{2,3,7,8\}}$ \\
	 \hline
	  $\mathcal{F}_{d_2, \{3,4,5,6\}}$ & $\mathcal{F}_{d_2, \{4,5,6,7\}}$ & $\mathcal{F}_{d_2, \{5,6,7,8\}}$ & $\mathcal{F}_{d_2, \{6,7,8,1\}}$ & $\mathcal{F}_{d_2, \{3,4,6,7\}}$ & $\mathcal{F}_{d_2, \{4,5,7,8\}}$ & $\mathcal{F}_{d_2, \{5,6,8,1\}}$ & $\mathcal{F}_{d_2, \{3,4,7,8\}}$ & $\mathcal{F}_{d_2, \{4,5,8,1\}}$ & $\mathcal{F}_{d_2, \{3,4,8,1\}}$\\
	 \hline
	 $\mathcal{F}_{d_3, \{4,5,6,7\}}$ & $\mathcal{F}_{d_3, \{5,6,7,8\}}$ & $\mathcal{F}_{d_3, \{6,7,8,1\}}$ & $\mathcal{F}_{d_3, \{7,8,1,2\}}$ & $\mathcal{F}_{d_3, \{4,5,7,8\}}$ & $\mathcal{F}_{d_3, \{5,6,8,1\}}$ & $\mathcal{F}_{d_3, \{6,7,1,2\}}$ & $\mathcal{F}_{d_3, \{4,5,8,1\}}$ & $\mathcal{F}_{d_3, \{5,6,1,2\}}$ & $\mathcal{F}_{d_3, \{4,5,1,2\}}$\\
	 \hline
	 $\mathcal{F}_{d_4, \{5,6,7,8\}}$ & $\mathcal{F}_{d_4, \{6,7,8,1\}}$ & $\mathcal{F}_{d_4, \{7,8,1,2\}}$ & $\mathcal{F}_{d_4, \{8,1,2,3\}}$ & $\mathcal{F}_{d_4, \{5,6,8,1\}}$ & $\mathcal{F}_{d_4, \{6,7,1,2\}}$ & $\mathcal{F}_{d_4, \{7,8,2,3\}}$ & $\mathcal{F}_{d_4, \{5,6,1,2\}}$ & $\mathcal{F}_{d_4, \{6,7,2,3\}}$ & $\mathcal{F}_{d_4, \{5,6,2,3\}}$ \\
	 \hline
	  $\mathcal{F}_{d_5, \{6,7,8,1\}}$ & $\mathcal{F}_{d_5, \{7,8,1,2\}}$ & $\mathcal{F}_{d_5, \{8,1,2,3\}}$ & $\mathcal{F}_{d_5, \{1,2,3,4\}}$ & $\mathcal{F}_{d_5, \{6,7,1,2\}}$ & $\mathcal{F}_{d_5, \{7,8,2,3\}}$ & $\mathcal{F}_{d_5, \{8,1,3,4\}}$ & $\mathcal{F}_{d_5, \{6,7,2,3\}}$ & $\mathcal{F}_{d_5, \{7,8,3,4\}}$ & $\mathcal{F}_{d_5, \{6,7,3,4\}}$\\
	 \hline
	  $\mathcal{F}_{d_6, \{7,8,1,2\}}$ & $\mathcal{F}_{d_6, \{8,1,2,3\}}$ & $\mathcal{F}_{d_6, \{1,2,3,4\}}$ & $\mathcal{F}_{d_6, \{2,3,4,5\}}$ & $\mathcal{F}_{d_6, \{7,8,2,3\}}$ & $\mathcal{F}_{d_6, \{8,1,3,4\}}$ & $\mathcal{F}_{d_6, \{1,2,4,5\}}$ & $\mathcal{F}_{d_6, \{7,8,3,4\}}$ & $\mathcal{F}_{d_6, \{8,1,4,5\}}$ & $\mathcal{F}_{d_6, \{7,8,4,5\}}$\\
	 \hline
	 $\mathcal{F}_{d_7, \{8,1,2,3\}}$ & $\mathcal{F}_{d_7, \{1,2,3,4\}}$ & $\mathcal{F}_{d_7, \{2,3,4,5\}}$ & $\mathcal{F}_{d_7, \{3,4,5,6\}}$ & $\mathcal{F}_{d_7, \{8,1,3,4\}}$ & $\mathcal{F}_{d_7, \{1,2,4,5\}}$ & $\mathcal{F}_{d_7, \{2,3,5,6\}}$ & $\mathcal{F}_{d_7, \{8,1,4,5\}}$ & $\mathcal{F}_{d_7, \{1,2,5,6\}}$ & $\mathcal{F}_{d_7, \{8,1,5,6\}}$\\
	 \hline
	 $\mathcal{F}_{d_8, \{1,2,3,4\}}$ & $\mathcal{F}_{d_8, \{2,3,4,5\}}$ & $\mathcal{F}_{d_8, \{3,4,5,6\}}$ & $\mathcal{F}_{d_8, \{4,5,6,7\}}$ & $\mathcal{F}_{d_8, \{1,2,4,5\}}$ & $\mathcal{F}_{d_8, \{2,3,5,6\}}$ & $\mathcal{F}_{d_8, \{3,4,6,7\}}$ & $\mathcal{F}_{d_8, \{1,2,5,6\}}$ & $\mathcal{F}_{d_8, \{2,3,6,7\}}$ & $\mathcal{F}_{d_8, \{1,2,6,7\}}$\\
	 \hline
	\end{tabular}}\\
\caption{\sl The index coding problem for ($N, K=8, L=2$)$-$CCDN with $M=2N/K$. } \label{ICPex}
\vspace{-0.5cm}
\end{table*}

\begin{table*}[t]
	\resizebox{\linewidth}{!}{
		\centering
		\begin{tabular}{|c|c|c|c|c|c|c|c|c|c|}
			\hline
			Column 1 & Column 2 & Column 3 & Column 4 & Column 5 & Column 6 & Column 7 & Column 8 & Column 9 & Column 10\\
			\hline
			\color{red}$\mathcal{F}_{d_1, \{2,3,4,5\}}$ & \color{red}$\mathcal{F}_{d_1, \{3,4,5,6\}}$ & \color{red}$\mathcal{F}_{d_1, \{4,5,6,7\}}$ & \color{red}$\mathcal{F}_{d_1, \{5,6,7,8\}}$ &  \color{red}$\mathcal{F}_{d_1, \{2,3,5,6\}}$ & \color{red}$\mathcal{F}_{d_1, \{3,4,6,7\}}$ & \color{red}$\mathcal{F}_{d_1, \{4,5,7,8\}}$ & \color{red}$\mathcal{F}_{d_1, \{2,3,6,7\}}$ & \color{red}$\mathcal{F}_{d_1, \{3,4,7,8\}}$ & \color{red}$\mathcal{F}_{d_1, \{2,3,7,8\}}$ \\
			\hline
			 &  &  & SI &  &  & SI &  & SI & SI \\
			 	\hline
			 &  & SI & SI &  & SI & SI & SI & SI & SI \\
			\hline
			& SI & SI & SI & SI & SI &  & SI &  &  \\
			\hline
		    SI & SI & SI & SI & SI & & SI &  &  &  \\
			\hline
			SI & SI & SI &  &  & SI & SI &  & SI &  \\
			\hline
			SI & SI &  &  & SI & SI &  & SI & SI & SI \\
			\hline
			SI &  &  &  & SI &  &  & SI &  & SI \\
			\hline
			\hline
			$(3,0,0)_2-$ICP & $(2,0,1)_2-$ICP & $(1,0,2)_2-$ICP & $(0,0,3)_2-$ICP  & $(2,1,0)_2-$ICP & $(1,1,1)_2-$ICP & $(0,1,2)_2-$ICP & $(1,2,0)_2-$ICP & $(0,2,1)_2-$ICP  & $(0,3,0)_2-$ICP \\
			\hline
	\end{tabular}}\\
\caption{\sl The side-information structure for Row 1 cells of the index coding problem for ($N, K=8, L=2$)$-$CCDN with $M=2N/K$ showed in Table \ref{ICPex}. }\label{siex}
\vspace{-0.5cm}
\end{table*}

Observe that in Column 1, the side-information of\footnote{Node ($m,n$) represents the cell at the $m^{th}$ row and $n^{th}$ column.}  Node (1,1) is of the form $\mathcal{K}_1=\{\text{Node}(b,1):b=1+\sum_{j=1}^{v}a_{w+2-j}+(v-1)L+r, v\in[w],r\in[L]\}$, where $w=2$, $a_1=0$, $a_2=0$, $a_3=3$. Because of symmetry, the side-information of the node ($j,1$) is of the form $\mathcal{K}_1=\{\text{node}(b,1):b=<j+\sum_{j=1}^{v}a_{w+2-j}+(v-1)L+r>_8, v\in[w],r\in[L]\}$. Therefore, Column 1 represents the $(3,0,0)_2-$ICP.  Similarly, Column 2 represents the $(2,0,1)_2-$ICP, Column 3 represents the $(1,0,2)_2-$ICP and so on. In the last row of Table \ref{siex}, we give the ICP form of the corresponding columns.

If we consider any column in Table \ref{siex}, the first node contains 4 side-information cells in $w=2$ chunks each of size $L=2$. Therefore, in each column, interference cells occur in $w+1=3$ chunks such that the total number of interference cells in a column  are $K-1-wL=3$, where $K-1$ is the total number of other users/nodes in a column excluding the first node  and $wL$ is the number of side-information nodes for the first node in a column. In Table \ref{siex}, observe that all possible combinations of interference chunk sizes ($b_1,b_2,b_3$) such that $b_i\geq 0$ and $b_1+b_2+b_3=K-1-wL=3$ are present across the columns.

Let $\mathcal{B}$ be the collection of all weak $w+1=3$ compositions of $K-wL-1=3$, i.e., $\mathcal{B}=\{({3,0,0})^T,(0,3,0)^T,(0,0,3)^T,$ $({2,1,0})^T,(0,2,1)^T, (1,0,2)^T,({2,0,1})^T,(1,2,0)^T,(0,1,2)^T,$\\ $({1,1,1})^T\}$. Note that for each element $(b_1.b_2,b_3)$ in $\mathcal{B}$, there exists a column with the $(b_1.b_2,b_3)-$ ICP in Table \ref{siex} and vice versa. Let $\mathcal{C}$ be the smallest subset of $\mathcal{B}$ such that the vectors $\mathbf{c }=$ $(c_{w+1},c_w,...,c_{1})^T$ in $\mathcal{C}$ are of the form  $0\leq c_j\leq c_{w+1} \forall j\in[w+1]$, and the vectors and their possible clockwise rotations (like $(c_{2},c_1,c_{w+1},...,c_{3})^T$, $(c_3,c_{2},c_1,c_{w+1},...,c_{4})^T$) will cover all the vectors in $\mathcal{B}$. For our ($N, K=8, L=2$)$-$CCDN with $M=2N/K$, $\mathcal{C}=\{({3,0,0})^T,({2,1,0})^T,$ $({2,0,1})^T,({1,1,1})^T\}$.  

Consider the vector $({3,0,0})^T$ in $\mathcal{C}$. Columns 1, 10 and 4 in Table \ref{siex} represent $(3,0,0)_2-$ICP, $(0,3,0)_2-$ICP and $(0,0,3)_2-$ICP respectively, i.e., the ICPs formed by the vector $(3,0,0)^T$ and its rotations. Therefore, Columns 1, 10 and 4 jointly represent  $\overline{(3,0,0)}_2-$ICP and the upper bound on the data transmission rate for the $\overline{(3,0,0)}_2-$ICP using Theorem  \ref{thm:icpub} is $2(K-wL)+w-1-3=6$. This can also be represented as $\sum_{\mathbf{b}\in \mathcal{B}_1}\frac{\min\{2(K-wL)+w-1-\mathbf{\widehat{b}},K\}}{3}$, where $\mathcal{B}_1=\{({3,0,0})^T,$ $(0,3,0)^T,$ $(0,0,3)^T\}$ and  $\mathbf{\widehat{b}}$ denotes the maximum component in the vector $\mathbf{b}$.

Similarly, consider the vectors $(2,1,0)^T$ and $(2,0,1)^T$  in $\mathcal{C}$. Columns 5, 9 and 3 in Table \ref{siex} jointly represent the  $\overline{(2,1,0)}_2-$ICP and the upper bound on the data transmission rate for  the $\overline{(2,1,0)}_2-$ICP using Theorem \ref{thm:icpub} is $2(K-wL)+w-1-3= 7$. This can also be represented as $\sum_{\mathbf{b}\in \mathcal{B}_2}\frac{\min\{2(K-wL)+w-1-\mathbf{\widehat{b}},K\}}{3}$, where $\mathcal{B}_2=\{({2,1,0})^T,$ $(0,2,1)^T,$ $(1,0,2)^T\}$. Columns 2, 8 and 7 in Table \ref{siex} jointly represent  the $\overline{(2,0,1)}_2-$ICP and the upper bound on the data transmission rate for the $\overline{(2,0,1)}_2-$ICP using Theorem  \ref{thm:icpub} is 7. Similarly, this can also be represented as $\sum_{\mathbf{b}\in \mathcal{B}_3}\frac{\min\{2(K-wL)+w-1-\mathbf{\widehat{b}},K\}}{3}$, where $\mathcal{B}_3=\{({2,0,1})^T,$ $(1,2,0)^T,$ $(0,1,2)^T\}$.

Consider the vector $(1,1,1)^T$ in $\mathcal{C}$. Column 6 represents the $(1,1,1)_2-$ICP. Unlike the above cases, note that all the rotations of $(1,1,1)^T$ represent the same ICP and hence, this case needs to be dealt separately. $(1,1,1)_2-$ICP is of the form $\widetilde{((1)\times3)}_2-$ICP mentioned in Appendix \ref{sec:maccproof}, where we characterize the optimal transmission rate of the ICPs of this form. The main idea is to split each file into sufficient number of equal sized subfiles to form an ICP of the form $\overline{(a_i,a_{i-1},...,a_1)}_L-$ICP and use Section \ref{ICP_results} results to get the bounds. In our case, we split each subfile in Column 6 to 3 equal sized sub-subfiles and form $\overline{(1,1,1)}_2-$ICP.  The details are given in Lemma \ref{lem:111_2ub} (Appendix \ref{sec:111_2ub}). From Theorem \ref{thm:icpub}, the upper bound on the transmission rate of the $\overline{(1,1,1)}_2-$ICP is 8. Since the subsubfiles size is 1/3 of the subfiles size, the upper bound on the transmission rate of the ${(1,1,1)}_2-$ICP is 8/3, which also agrees with Lemma \ref{partialub} in Appendix \ref{sec:maccproof}.  
	This can also be represented as $\sum_{\mathbf{b}\in \mathcal{B}_4}\frac{\min\{2(K-wL)+w-1-\mathbf{\widehat{b}},K\}}{3}$, where $\mathcal{B}_4=\{(1,1,1)^T\}$.

Note that the data transmission rate for the ICP given in Table \ref{ICPex} is upper bounded by the sum of the data transmission rates of $\overline{(3,0,0)}_2-$ICP, $\overline{(2,1,0)}_2-$ICP, $\overline{(2,0,1)}_2-$ICP, and $(1,1,1)_2-$ICP, which is upper bounded by 6+7+7+8/3=68/3. This upper bound can be written as $\sum_{j=1}^{4}\sum_{\mathbf{b}\in \mathcal{B}_j}\frac{\min\{2(K-wL)+w-1-\mathbf{\widehat{b}},K\}}{3}$. Since $\mathcal{B}$ is the disjoint union of $\mathcal{B}_1$, $\mathcal{B}_2$, $\mathcal{B}_3$ and $\mathcal{B}_4$, $\sum_{j=1}^{4}\sum_{\mathbf{b}\in \mathcal{B}_j}\frac{\min\{2(K-wL)+w-1-\mathbf{\widehat{b}},K\}}{3}$ is equal to $\sum_{\mathbf{b}\in \mathcal{B}}\frac{\min\{2(K-wL)+w-1-\mathbf{\widehat{b}},K\}}{3}$.

We now use this upper bound for the ICP in Table \ref{ICPex} to give an upper bound on the server transmission rate in the $(N, K = 8, L=2)-$CCDN at $M = 2N/K$. Till now, we calculated the upper bound on the ICP rate assuming a unit subfile size. Recall that, in the placement phase, we divide each file into 20 subfiles, each of size 1/20 units. Therefore, for ($N, K=8, L=2$)$-$CCDN at memory point $M=2N/K$, $R_\text{new}(2N/K)\leq \frac{1}{20} \times \sum_{\mathbf{b}\in \mathcal{B}}\frac{\min\{2(K-wL)+w-1-\mathbf{\widehat{b}},K\}}{3} =68/60$ units, which is smaller than the other known upper bounds:  2 units (see \eqref{eq:oldub}) in \cite{reddy2020rate},  5/3 units in \cite{sasi2020improved}, 4/3 units (see \eqref{eqn:colorub}) in \cite{hachem2017codedmulti} and 7/6 units in \cite{serbetci2019multi}.
\end{example}

\subsection{Closed form expression for ($N,K,L\geq K/2$)$-$CCDN at $M=N/K$}\label{sec:closed}
Theorem  \ref{thm:multiub2} says that for an ($N,K,L$)$-$CCDN, at memory point $M=iN/K$, $i\in[\lfloor K/L\rfloor]$, we achieve $R_{\text{new}}(M).$ But, the expression for $R_{\text{new}}(M)$ given in \eqref{eqn:multiub2} isn't in closed form. However, for the subclass of ($N,K,L\geq K/2$)$-$CCDN at memory point $M=N/K$, we are able to evaluate the expression explicitly and provide a simple upper bound on the server transmission rate.  Note that $M = N/K$ is the only non-trivial corner point when $L\ge K/2$ since the remaining two corner points are $M=0$ and $M=2N/K$ where its easy to achieve the rates $K$ and 0 respectively. 
\begin{corollary}\label{cor:closed}
	For an ($N,K,L\geq K/2$)$-$CCDN, let $R^*(M)$ be the optimal transmission rate under the restriction of uncoded placement at cache size $M$. Then, at memory point $M=N/K$,
	$$R^*(M)\leq R_{\text{new}}(M)\leq \frac{5(K-L)(K-L+1)}{8K}.$$
\end{corollary}
The proof of the above corollary  is given in Appendix \ref{sec:closedproof}.
\begin{remark}
	The sub-class of ($N,K,L\geq K/2$)$-$CCDN was also studied in \cite{reddy2020rate} and at memory point $M=N/K$,  the achievable rate in Corollary \ref{cor:closed} is $\approx$ 5/8 times smaller compared to the achievable rate in \cite[Corollary 2]{reddy2020rate}. By comparing our result with the lower bound \cite[Theorem 3]{reddy2020rate}, we can say that $R_{\text{new}}(M)/R^*(M)\leq 5/4$, i.e., the best-known multiplicative gap between the achievable rate-memory trade-off and optimal rate-memory trade-off under the restriction of uncoded placement policies is reduced from 2 \cite{reddy2020rate} to 5/4, using our new results.
\end{remark}
\begin{remark}
	A key difficulty in getting a closed form expression for the general $(N,K,L)$-CCDN achievable rate is getting a clean characterization of the number of weak $m$ compositions of $n$  for which the largest value $\mathbf{\widehat{b}}$ is equal to some given integer $t$.  A similar question addressed in \cite{stanley2011enumerative} as one of the exercises. It states that under the restriction that the largest element $\mathbf{\widehat{b}}$ in the composition $\mathbf{b}$ be smaller than $t$, the number of possible $m$ weak compositions of $n$ are 
	$$\sum_{r,s\in\mathbb{N}^+:r+ts=n}(-1)^s{m \choose s}{m+r-1 \choose r},$$
	where $\mathbb{N}^+=\mathbb{N}\cup \{0\}.$ 
	Therefore, if the largest element in the composition $\widehat{\mathbf{b}}$ is equal to $t$, then the number of possible weak $m$ compositions of $n$ are 
	$$\sum_{r,s\in\mathbb{N}^+:r+(t+1)s=n}(-1)^s{m \choose s}{m+r-1 \choose r}- \sum_{r,s\in\mathbb{N}^+:r+ts=n}(-1)^s{m \choose s}{m+r-1 \choose r}.$$
	We can plug in the above expression into \eqref{eqn:multiub2} to get a semi-closed form expression for the general ($N,K,L$)$-$CCDN achievable rate.
\end{remark}

\section{Numerical Results}\label{sec:numerical}
\begin{figure}[t]
	\centering
	\vspace{-0.2in}
		\centering
		\begin{minipage}{0.45\textwidth}
			\centering
			\includegraphics[width = 1.0\textwidth]{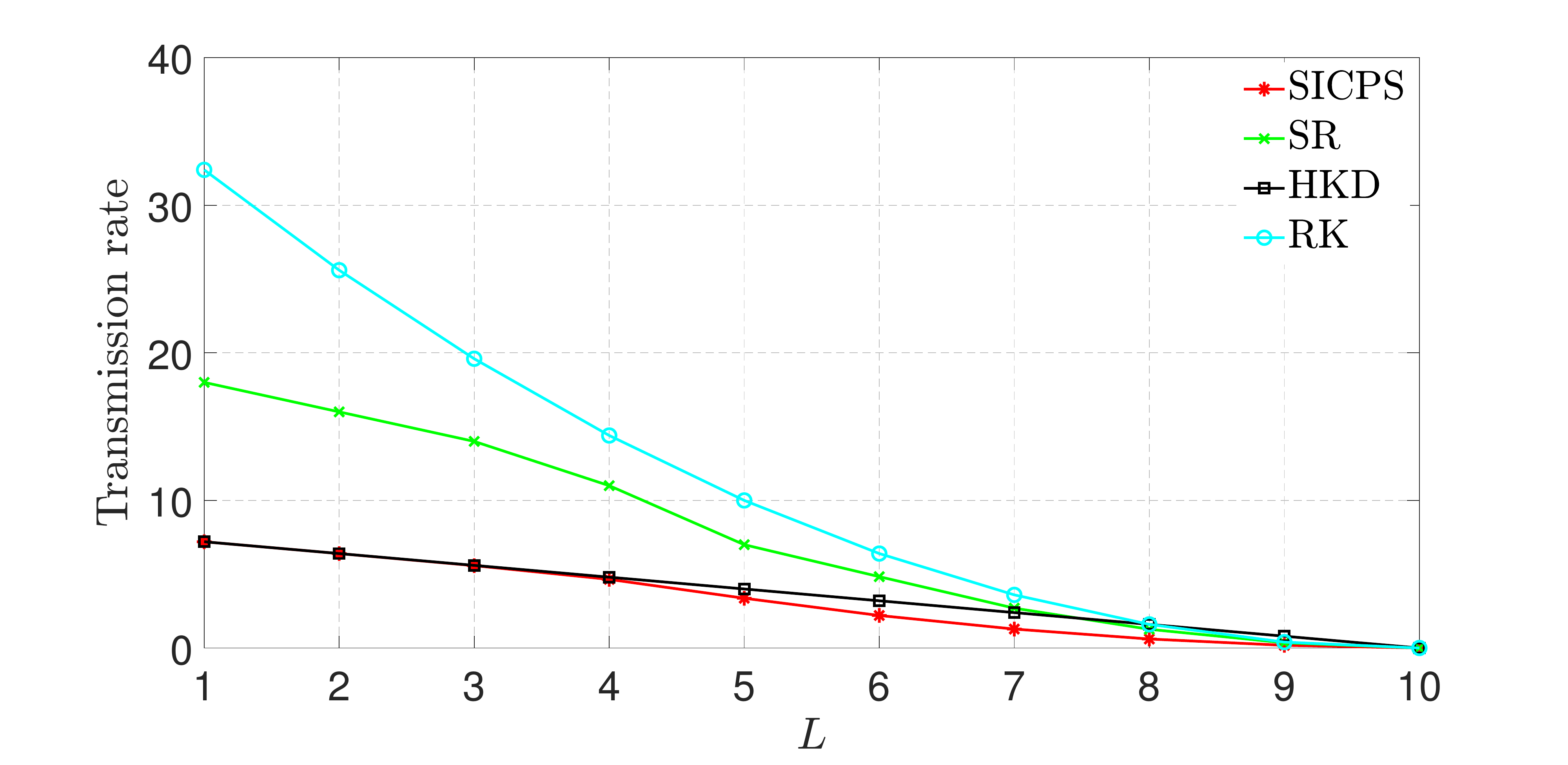}\\ (i)
		\end{minipage}
		\begin{minipage}{0.45\textwidth}
			\centering
			\includegraphics[width = 1.0\textwidth]{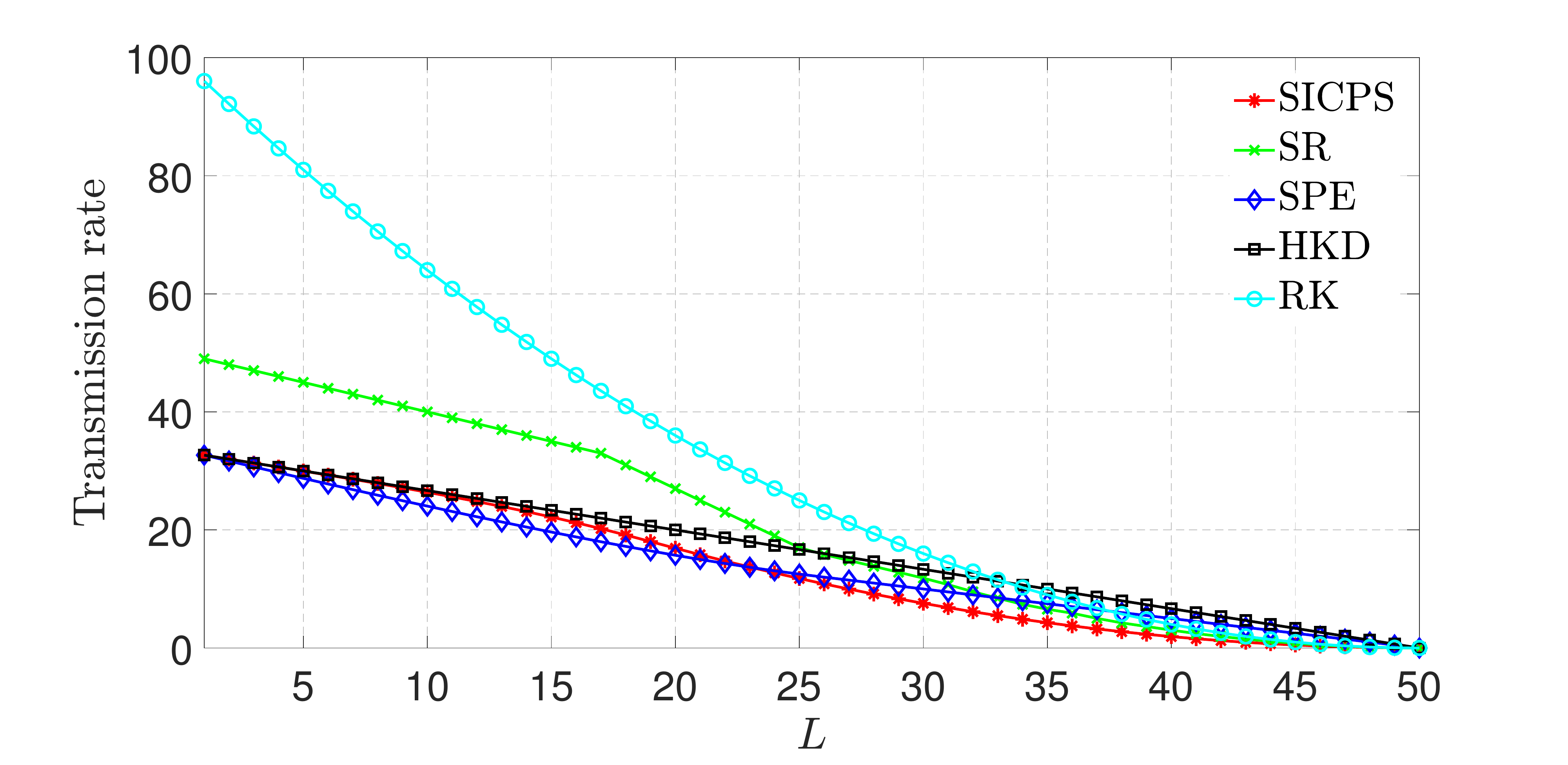} (ii)
		\end{minipage}
		\caption{\sl Plot of the transmission rate $R$ as a function of the  multi-access parameter $L$ for  (i) the ($N=100, K=40, L$)-CCDN at memory point $M=4N/K=10$ and (ii) the ($N=100, K=100, L$)-CCDN at memory point $M=2N/K=2$. }\label{fig:asl} 
		\vspace{-0.15in} 
\end{figure}
In Corollary \ref{cor:comparison}, we have shown that the rate achieved by
our proposed scheme is always upper bounded by the achievable rates in \remove{our multi-access results are better than the results in} \cite{reddy2020rate,hachem2017codedmulti}. Since the expression  for $R_{\text{new}}(M)$ given in \eqref{eqn:multiub2} is not in closed form, it is not easy to  analytically compare our results with the other works \cite{serbetci2019multi,sasi2020improved}.  Example  \ref{ex:multiaccess} showed that the performance of our proposed scheme can be better than the achievable rates in \cite{serbetci2019multi,sasi2020improved}, and in this section we conduct numerical evaluations to compare our results with prior works \cite{reddy2020rate,hachem2017codedmulti,serbetci2019multi,sasi2020improved} more broadly.
 We label the achievable rates derived in \cite{hachem2017codedmulti,cheng2020novel} as \textit{ `HKD',} \cite{reddy2020rate,serbetci2019multi,sasi2020improved} as \textit{`RK', `SR', `SPE'} respectively, and our new results based on the structured index coding problem solution as \textit{`SICPS'}. 

In Figure \ref{fig:asl}, we plot the transmission rate as a function of the access degree $L$, for (i) the ($N=100, K=40, L$)-CCDN at memory point $M=4N/K=10$ and (ii) the ($N=100, K=100, L$)-CCDN at memory point $M=2N/K=2$. As expected, the transmission rate is reducing as  $L$ increases.  The figures shows that our achievable rate is better than the achievable rates \textit{`RK', `HKD', `SR'} derived in \cite{reddy2020rate,hachem2017codedmulti,sasi2020improved}  for all values of $L$. It is also better than the achievable rate \textit{`SPE'} proposed in \cite{serbetci2019multi} for larger values of $L$.  Also, note that the results in \cite{serbetci2019multi} are only applicable when $M = 2N/K$ and $M={(K-1)N}/{KL}$. 

	\begin{figure}[t]
	\centering
	\centering
	\begin{minipage}{0.45\textwidth}
		\centering
		\includegraphics[width = 1.0\textwidth]{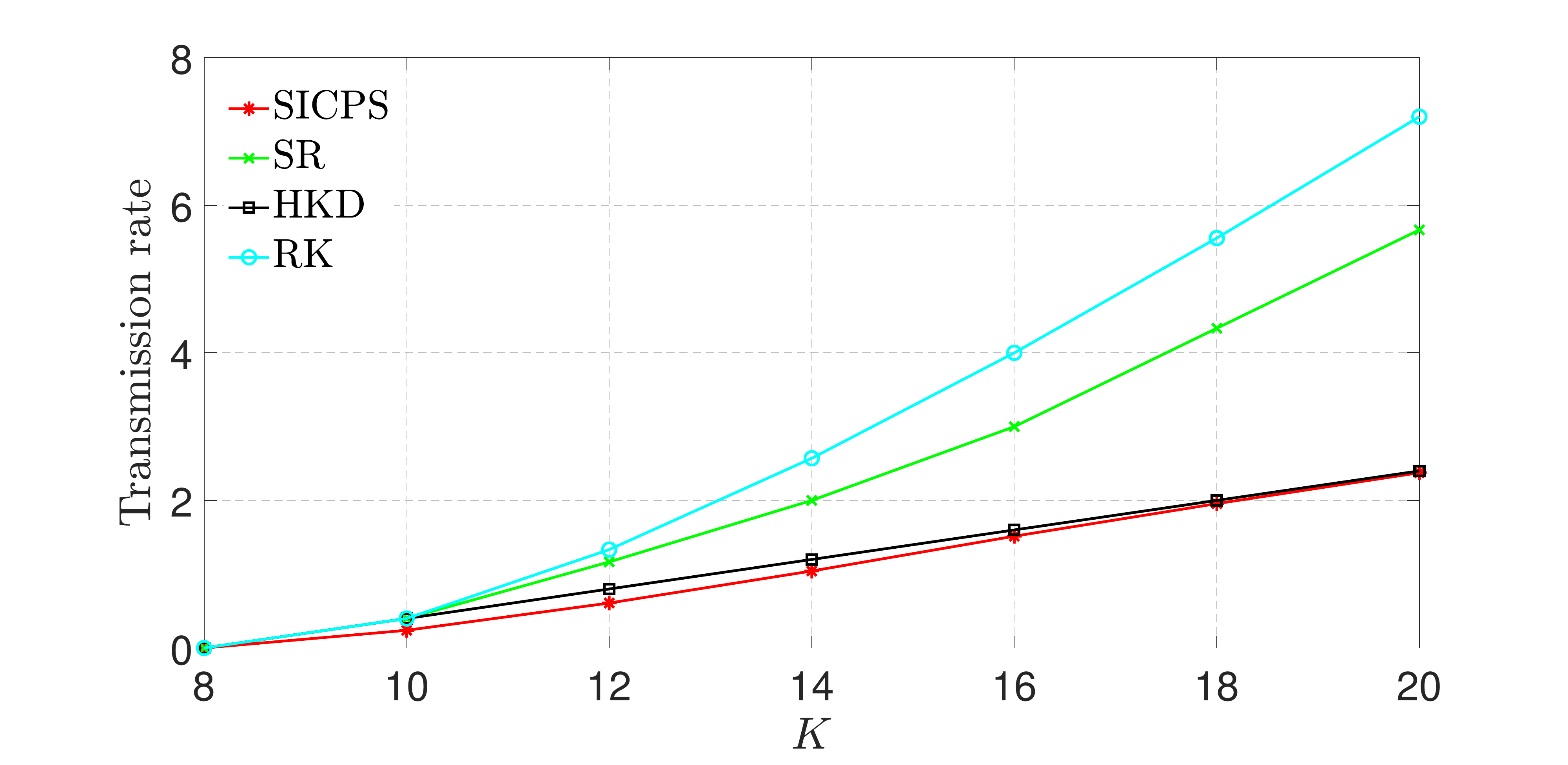}\\ (i)
	\end{minipage}
	\begin{minipage}{0.45\textwidth}
		\centering
		\includegraphics[width = 1.0\textwidth]{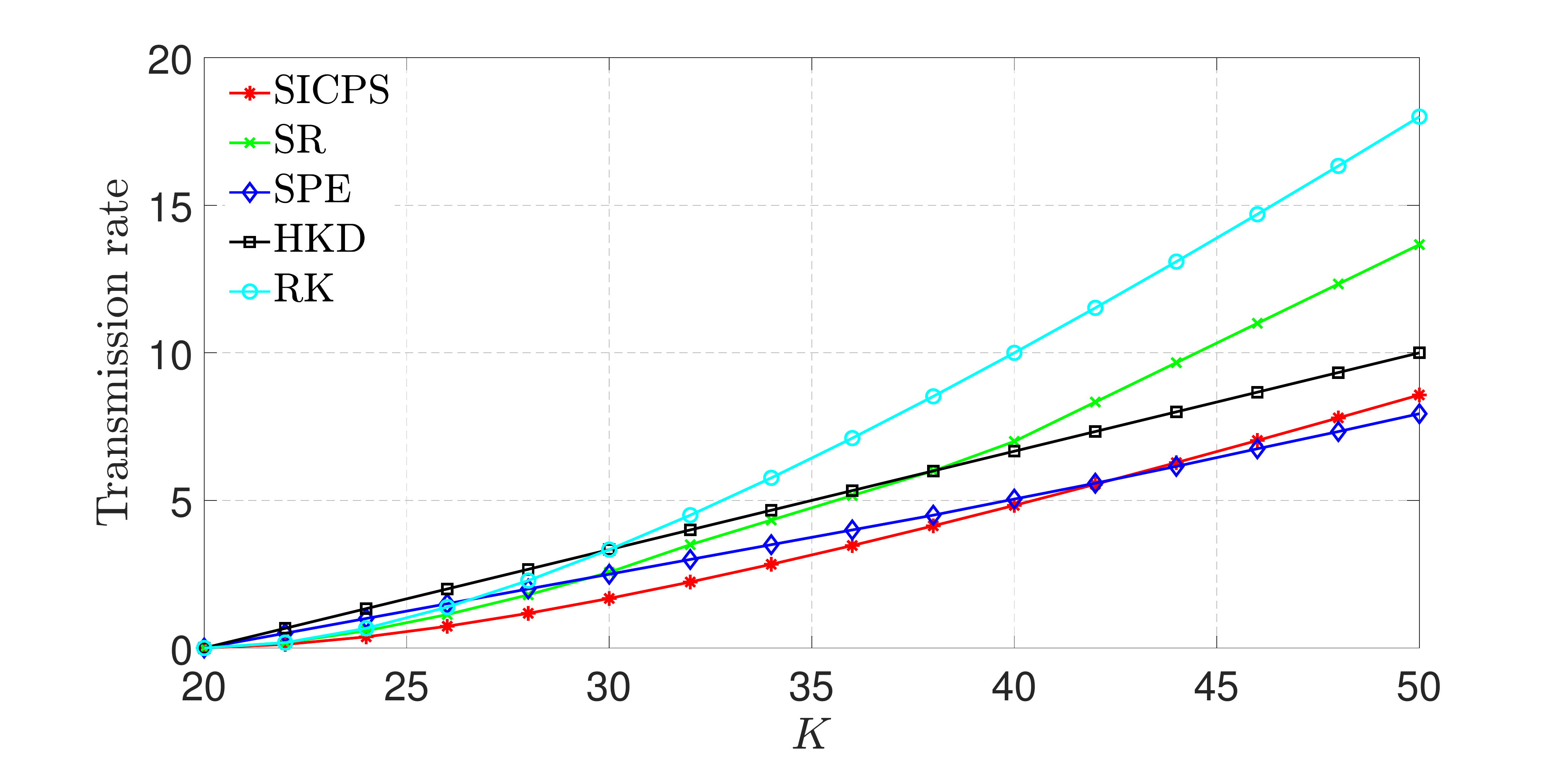} (ii)
	\end{minipage}
	\caption{\sl Plot of the transmission rate $R$ as a function of the  number of users $K$ for (i) the ($N=2K, K, L=2$)-CCDN at memory point $M=4N/K=8$ and (ii) the ($N=K, K, L=10$)-CCDN at memory point $M=2N/K=2$. }\label{fig:ask} 
	\vspace{-0.1in} 
\end{figure}

In Figure \ref{fig:ask}, we plot the transmission rate as a function of the  number of users $K$, for (i) the ($N=2K, K, L=2$)-CCDN at memory point $M=4N/K=8$ and (ii) the ($N=K, K, L=10$)-CCDN at memory point $M=2N/K=2$. As expected, the transmission rate is increasing as  $K$ increases. The figure shows that our achievable rate is better than the achievable rates \textit{`RK', `HKD', `SR'} derived in \cite{reddy2020rate,hachem2017codedmulti,sasi2020improved}  for all values of $K$.  It is also better than the achievable rate \textit{`SPE'} proposed in \cite{serbetci2019multi} for smaller values of $K$, where $L$ is comparable with $K$. 

\begin{figure}[t]
	\centering
	\vspace{-0.15in}
	\begin{minipage}{0.45\textwidth}
	\centering
	\includegraphics[width = 1.0\textwidth]{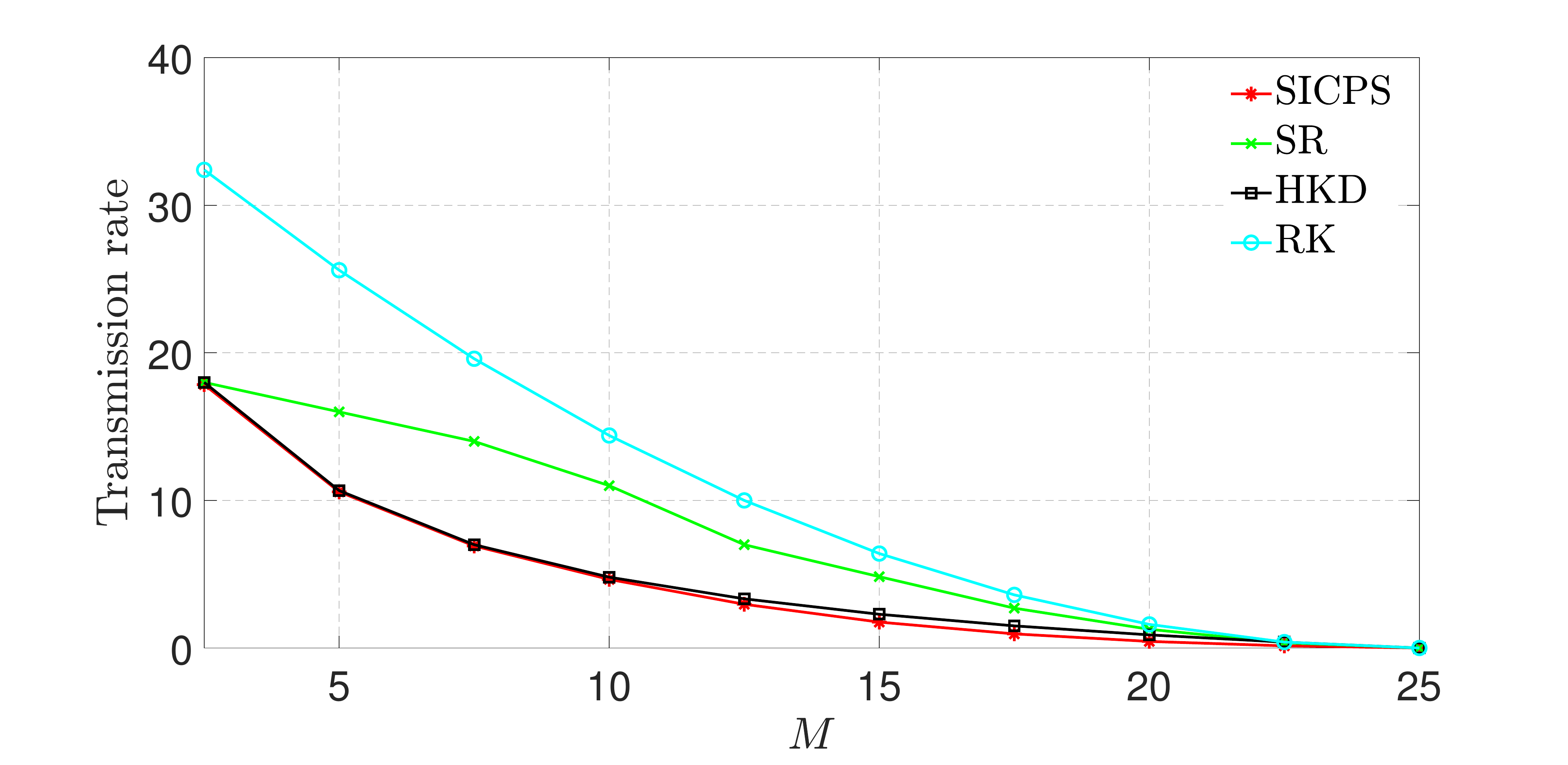}\\ (i)
	\end{minipage}
\begin{minipage}{0.45\textwidth}
	\centering
	\includegraphics[width = 1.0\textwidth]{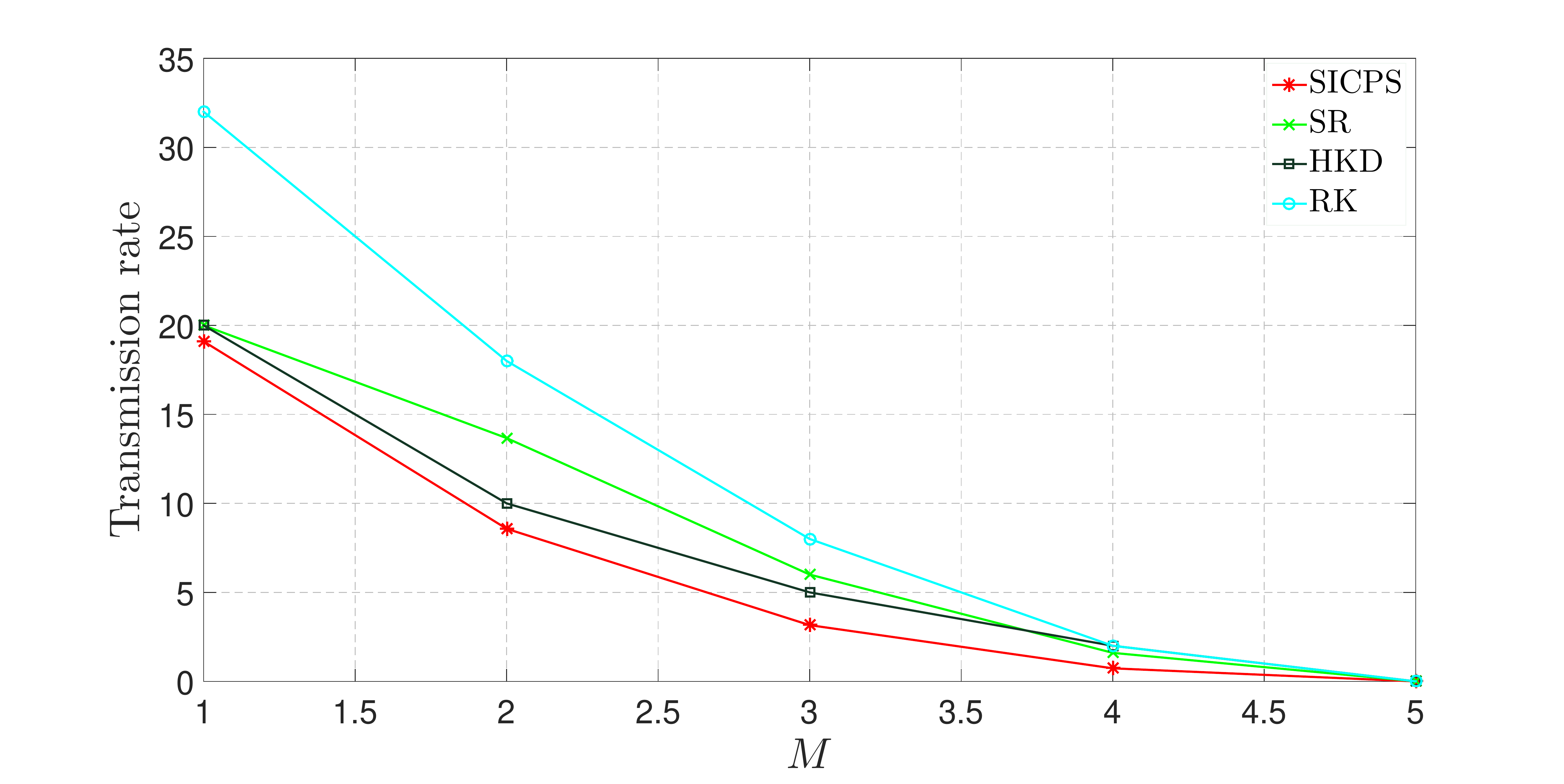} (ii)
\end{minipage}
	\caption{\sl Plot of the transmission rate $R$ as a function of the  cache size $M$ for (i) the ($N=100, K=40, L=4$)-CCDN and (ii) the ($N=50, K=50, L=10$)-CCDN. }\label{fig:asm} 
	\vspace{-0.1in} 
\end{figure}
In Figure \ref{fig:asm}, we plot the transmission rate as a function of the  cache size $M$, for (i) the ($N=100, K=40, L=4$)-CCDN and (ii) the ($N, K=50, L=10$)-CCDN. As expected, the transmission rate decreases as  $M$ increases. The figure shows that our achievable rate is better than the achievable rates \textit{`RK', `HKD', `SR'} derived in \cite{reddy2020rate,hachem2017codedmulti,sasi2020improved}  for all values of $M$. 
 Since the results of \cite{serbetci2019multi} only hold for $M = 2N/K$ and $M = (K-1)N/KL$, we haven't included them in this plot. 


\section{Conclusions}\label{sec:conclusions}
In this work, we studied a new class of structured index coding problems (ICPs) formed by the union of several symmetric ICPs and characterized its optimal server transmission rate up to a constant factor. Several variants of such unions are of interest and techniques similar to those used in this paper will be helpful in analysing them as well. Identifying interesting classes of such ICPs and deriving similarly tight characterizations of the optimal rates is certainly a direction we plan to pursue.  

Using our ICP results, we were able to derive improved achievable rates for the general Multi-access Coded Caching (MACC) problem. The general expression is not in closed form though and we were able to simplify it only for large $L$. Obtaining simpler bounds and also deriving improved lower bounds on the optimal rate-memory trade-off of the general MACC problem are also an important part of our future work.   

\bibliographystyle{IEEEtran}
\bibliography{myref2}


\appendix \label{sec:icp_proofs}


\subsection{Proof of Theorem \ref{thm:icpub}}\label{sec:icpub}
 We use the following lemmas to prove Theorem \ref{thm:icpub} and their proofs are relegated to Appendices \ref{sec:interference} and \ref{sec:proper}.
 \begin{lemma}\label{lemma:interference}
 	In Table \ref{Tab:genicp}, for any User $k$, $n\in[1:i-1]$, and $m\in[<a_n+L+1>_K:<K-1>_K]$, 
 	\begin{itemize}
 		\item if Node $(<k+m>_K,n+1)$ is a side-information node then Node $(<k+m-a_n-L>_K,n)$ is also a side-information node,
 		\item if Node $(<k+m>_K,n+1)$ is an interference node then Node $(<k+m-a_n-L>_K,n)$ is also an interference node.
 	\end{itemize} 
 \end{lemma}

\begin{lemma}\label{lemma:proper}
	For the $\overline{(a_i,a_{i-1},...,a_1)}_L-$ICP given in Table \ref{Tab:genicp}, if we assign Color $c$ to nodes ($<c+\sum_{m=1}^{v-1}a_m+(v-1)L>_K,v$) $\forall v \in[i]$, then the coloring scheme is a \textit{proper coloring} scheme.
\end{lemma}

%
\begin{proof}{\bf \hspace{-0.15in}Proof of Theorem \ref{thm:icpub}:}
Recall that, in an $\overline{(a_i,a_{i-1},...,a_1)}_L-$ICP, each user wants $i$ files and has $i(i-1)L$ files as side-information. We form a $K\times i$ table such that $k^{th}$ row and $j^{th}$ column contains $j^{th}$ file requested by User $k$, see Table \ref{Tab:genicp}. We denote the $k^{th}$ row and $j^{th}$ column's entry as Node $(k,j)$.

\remove{\begin{table}[h]
		\vspace{-.25cm}
		\centering
			\begin{tabular}{| c | c | c | c |}
				\hline 
				$x_{1,1}$ &  $x_{1,2}$ & ... & $x_{1,i}$\\
				\hline 
				$x_{2,1}$ &  $x_{2,2}$ & ... & $x_{2,i}$\\
				\hline 
				\vdots & \vdots &\vdots &\vdots \\
				\hline
				$x_{K,1}$ &  $x_{K,2}$ & ... & $x_{K,i}$\\
				\hline 
		\end{tabular}
		\vspace{0.2cm}
		\caption{$\overline{(a_i,a_{i-1},...,a_1)}_L-$ICP.  } \label{Tab:icp}
		\vspace{-.5cm}
\end{table}}
 Note that the $\overline{(a_i,a_{i-1},...,a_1)}_L-$ICP is a UICP with $K$ users, and each user requesting $i$ distinct files. As we did for Example \ref{example210} in Section \ref{ICP_results}, we convert this UICP into an SUICP with $i\times K$ virtual users, each one requesting  a distinct file. In particular, each user in the UICP is  mapped to $i$ virtual users in the SUICP, such that each virtual user requesting a distinct file of the original user's requested files.  The side-information at the virtual user is the same as 
its corresponding original user. In Table \ref{Tab:genicp}, recall that we use Node $(k,j)$ to represent the $j^{th}$ file requested by User $k$ and refer to it as the $j^{th}$ virtual user of User $k$. The side-information for Node $(k,j)$ in the constructed SUICP is the same as the side-information for User $k$ in the original UICP.  
We give an upper bound for the SUICP, and it also works as a bound for the original UICP.

Recall from Section \ref{sec:ICP_preliaries} that for an SUICP, a proper coloring scheme assigns a color to each user such that no user shares its color with any of its interfering users.
 We take $K$ colors and assign Color $c$ to nodes $(<c+\sum_{m=1}^{v-1}a_{m}+(v-1)L>_K,v)$ $\forall v\in[i]$. Note that, every color occurs in a column exactly once. Lemma \ref{lemma:proper} tells that the above coloring scheme is a proper coloring scheme.

 The closed anti-outneighborhood of User $k$  is defined as the set containing User $k$ itself and all its interfering users. The local chromatic number $\mathcal{X}_l$ of an ICP is defined as
 the maximum number of different colors that appear in the closed anti-outneighborhood of any user, minimized over all proper coloring schemes.
 
 
 Consider a Node ($k,j$) and calculate the number of colors in closed anti-outneighborhood of it by going over column by column, i.e., we  first add the first column's colors contributing to the closed  anti-outneighborhood, then we add second column's colors contributing to the closed  anti-outneighborhood, then we add third column's colors and so on finally $i^{th}$ column's colors. Recall that, in the first column, we assign $K$ colors with one color per node, and  Node ($k,j$) contains $(i-1)L$ side-information nodes which are not involved under the closed anti-outneighborhood. Therefore, there are $K - (i-1)L$ colors in the first column which contribute to the closed anti-outneighborhood of Node($k,j$). Note that according to our coloring scheme, Node $(n,2)$ and Node $(<n-a_1-L>_K,1)$ share a common color. From Lemma \ref{lemma:interference}, $\forall m\in[<a_1+L+1>_K:<K-1>_K]$, if Node ($<k+m>_K,2$) is an interference node then Node ($<k+m-a_1-L>_K,1)$ is also an interference node.  Therefore, using Lemma \ref{lemma:interference} and our coloring scheme, we can see that a color  assigned to any interference node $(<k+m>_K,2)$ of Node ($k,j$),  $m\in[a_1+L+1:K-1]$ is also assigned to an interference node $(<k+m-a_1-L>_K,1)$ of Node ($k,j$) in the first column and hence is already added to the set. Therefore, the nodes $(<k+m>_K,2)$,  $m\in[a_1+L+1:K-1]$ don't add any new colors to the set. According to our problem setup, amongst the remaining nodes in Column 2 given by  $(<k+m>_K,2)$,  $m\in[0:a_1+L]$,  there are $L$ side-information nodes.   Therefore, the extra colors added from Column 1 to Column 2 are at most $a_1+1$. Similarly, for a general $p\in[i-1]$, we can see that  Node $(n,p+1)$ and Node $(<n-a_p-1>_K,p)$ share a common color. From Lemma \ref{lemma:interference}, $\forall m\in[<a_p+L+1>_K:<K-1>_K]$, if Node ($<k+m>_K,p+1$) is an interference node then Node ($<k+m-a_p-L>_K,p)$ is also an interference node. Therefore, using Lemma \ref{lemma:interference} and our coloring scheme, we can see that a color  assigned to any interference node $(<k+m>_K,p+1)$ of Node ($k,j$),  $m\in[a_p+L+1:K-1]$ is also assigned to an interference node $(<k+m-a_p-L>_K,p)$ of Node ($k,j$) in the $p^{th}$ column and hence is already added to the set. Therefore, the nodes $(<k+m>_K,p+1)$,  $m\in[a_p+L+1:K-1]$ don't add any new colors to the set. According to our problem setup, amongst the remaining nodes in Column $p+1$ given by $(<k+m>_K,p+1)$,  $m\in[0:a_p+L]$,  there are $L$ side-information nodes. Therefore, the extra colors added from Column $p$ to Column $p+1$ are at most $a_p+1$.
 	 Therefore, the total number of colors in the closed anti-outneighborhood of any Node ($k,j$) is less than or equal to $$K-(i-1)L+\sum_{m=1}^{i-1}(a_m+1)= 2(K-(i-1)L)-a_i+i-2.$$
  Therefore, from Lemma \ref{lemma:icpupperbound}, $$R^*\leq \mathcal{X}_l\leq 2(K-(i-1)L)+i-2-a_i.$$
 
 In our coloring scheme, we are using only $K$ colors. Hence, every node contains at most $K$ colors in its closed anti-outneighborhood. Therefore, from Lemma \ref{lemma:icpupperbound}, we also have $$R^*\leq\mathcal{X}_l\leq K.$$ This completes the proof.
\end{proof}
\subsection{Proof of Theorem \ref{thm:icplb}}\label{sec:icplb}
As an illustration,  we first discuss the $\overline{(5,3,2)}_2-$ICP. Recall from Section \ref{sec:icp_setting} that the number of nodes in a $\overline{(5,3,2)}_2-$ICP is given by $K=(i-1)L+\sum_{j=1}^{i}a_j+1=15$ and that the $\overline{(5,3,2)}_2-$ICP can be represented by a  $15\times 3 $ table such that $k^{th}$ row and $j^{th}$ column contains $j^{th}$ file requested by User $k$, see Table \ref{Tab:icplb13e}. We denote the $k^{th}$ row and $j^{th}$ column's entry as Node $(k,j)$. In Table \ref{Tab:icplb13e}, we also highlight User 1's requested files with red color fonts and User 1's side-information files with shaded cells and blue-colored fonts.
\begin{table}[h]
	\begin{minipage}{.49\linewidth}
		\centering
		\begin{tabular}{| c | c | c |}
			\hline 
			\color{red}$\mathbf{x_{1,1}}$ &  \color{red}$\mathbf{x_{1,2}}$ &\color{red}$\mathbf{x_{1,3}}$ \\
			\hline 
			$x_{2,1}$ &  $x_{2,2}$ & $x_{2,3}$\\
			\hline 
			$x_{3,1}$ &  $x_{3,2}$ &  $x_{3,3}$\\
			\hline 
			$x_{4,1}$ & \cellcolor{gray} \color{blue} $x_{4,2}$ & $x_{4,3}$\\
			\hline 
				${x_{5,1}}$ & \cellcolor{gray} \color{blue}${x_{5,2}}$ & \cellcolor{gray} \color{blue}$x_{5,3}$\\
			\hline
				${x_{6,1}}$ &  $x_{6,2}$ & \cellcolor{gray} \color{blue}$x_{6,3}$\\
			\hline
			\cellcolor{gray} \color{blue}$x_{7,1}$ &  $x_{7,2}$ & $x_{7,3}$ \\
			\hline 
			\cellcolor{gray} \color{blue}$x_{8,1}$ &  $x_{8,2}$ & $x_{8,3}$\\
			\hline
			$x_{9,1}$ &  $x_{9,2}$ & \cellcolor{gray} \color{blue}$x_{9,3}$\\
			\hline
			$x_{10,1}$ &  $x_{10,2}$ & \cellcolor{gray} \color{blue}$x_{10,3}$\\
			\hline
			$x_{11,1}$ &  \cellcolor{gray} \color{blue}$x_{11,2}$ & $x_{11,3}$\\
			\hline
			\cellcolor{gray} \color{blue}$x_{12,1}$ &  \cellcolor{gray} \color{blue}$x_{12,2}$ & $x_{12,3}$\\
			\hline
			\cellcolor{gray} \color{blue}$x_{13,1}$ &  $x_{13,2}$ & $x_{13,3}$\\
			\hline
			$x_{14,1}$ &  $x_{14,2}$ & $x_{14,3}$\\
			\hline
			$x_{15,1}$ &  $x_{15,2}$ & $x_{15,3}$\\
			\hline
		\end{tabular}
		\vspace{0.2cm}
		\caption{\sl $\overline{(5,3,2)}_2-$ICP.  Here, we highlight the requested files of User 1 with the red color fonts and the side-information files of User 1 with the shaded cells and the blue color fonts. The remaining files are the interference files for User 1.  } \label{Tab:icplb13e}
	\end{minipage}
	\begin{minipage}{.49\linewidth}
		\centering
		\begin{tabular}{| c | c | c |}
			\hline 
		\cellcolor{green}	\color{red}$\mathbf{x_{1,1}}$ &  \cellcolor{green} \color{red}$\mathbf{x_{1,3}}$ &\cellcolor{green} \color{red}$\mathbf{x_{1,2}}$ \\
			\hline 
		\cellcolor{green}	$x_{2,1}$ & \cellcolor{green} $x_{2,3}$ &\cellcolor{green} $x_{2,2}$\\
			\hline 
		\cellcolor{green}	$x_{3,1}$ &\cellcolor{green}  $x_{3,3}$ & \cellcolor{green} $x_{3,2}$\\
			\hline 
		\cellcolor{green}	$x_{4,1}$ & \cellcolor{green}  $x_{4,3}$ & \cellcolor{gray}\color{blue}$x_{4,2}$\\
			\hline 
		\cellcolor{green}	${x_{5,1}}$ & \cellcolor{gray} \color{blue}${x_{5,3}}$ & \cellcolor{gray} \color{blue}$x_{5,2}$\\
			\hline
		\cellcolor{green}	${x_{6,1}}$ & \cellcolor{gray} \color{blue} $x_{6,3}$ &  $x_{6,2}$\\
			\hline
			\cellcolor{gray} \color{blue}$x_{7,1}$ &  $x_{7,3}$ & $x_{7,2}$ \\
			\hline 
			\cellcolor{gray} \color{blue}$x_{8,1}$ &  $x_{8,3}$ & $x_{8,2}$\\
			\hline
			$x_{9,1}$ & \cellcolor{gray} \color{blue} $x_{9,3}$ & $x_{9,2}$\\
			\hline
			$x_{10,1}$ & \cellcolor{gray} \color{blue} $x_{10,3}$ & $x_{10,2}$\\
			\hline
			$x_{11,1}$ &  $x_{11,3}$ & \cellcolor{gray} \color{blue}$x_{11,2}$\\
			\hline
			\cellcolor{gray} \color{blue}$x_{12,1}$ &  $x_{12,3}$ & \cellcolor{gray} \color{blue} $x_{12,2}$\\
			\hline
			\cellcolor{gray} \color{blue}$x_{13,1}$ &  $x_{13,3}$ & $x_{13,2}$\\
			\hline
			$x_{14,1}$ &  $x_{14,3}$ & $x_{14,2}$\\
			\hline
			$x_{15,1}$ &  $x_{15,3}$ & $x_{15,2}$\\
			\hline
		\end{tabular}
		\vspace{0.2cm}
		\caption{\sl Rearranged $\overline{(5,3,2)}_2-$ICP. Here, the cells highlighted with the green color contributes to the acyclic induced subgraph. Observe that the columns are shuffled so that the number of interference nodes in the first chunk for User 1 in each given column are in non-increasing order. } \label{Tab:icplb3e}
	\end{minipage}%
\vspace{-0.5cm}
\end{table}

Note that the $\overline{(5,3,2)}_2-$ICP is a UICP with 15 users, and each user requests 3 files. Similar to the proof of Theorem \ref{thm:icpub}, we convert this UICP into an SUICP with 45 virtual users, each one requesting a distinct file. In particular, each user in the UICP is  mapped to 3 virtual users in the SUICP with each virtual user requesting a distinct file from amongst the original user's requested files.  The side-information at the virtual user is the same as its corresponding original user. In Table \ref{Tab:icplb13e}, we use Node $(k,j)$ to represent the $j^{th}$ virtual user of User $k$. The optimal broadcast rate in the two ICPs (the virtual SUICP and the original UICP) are equal. We give a lower bound for the SUICP, and it also works as a bound for the original UICP. 

Recall that an SUICP with $45$ users can equivalently be represented by a side information graph $\mathcal{G}$ with $45$ nodes such that each node represents a unique user and there exists an edge from Node $k$ to Node $j$ if  User $j$'s requested file belongs to  \textit{Known-set} $\mathcal{K}_k$ of User $k$. 
Lemma \ref{lemma:icplowerbound} shows that the number of nodes in any  acyclic induced sub-graph of the side-information graph $\mathcal{G}$ is  a lower bound on the broadcast rate of the corresponding ICP.

Now, we construct an acyclic induced sub-graph for our ICP.  The idea is: first we take all nodes belonging to the same row of the ICP table, which do not contain any cycles by definition. Then, we increase the set size step by step by adding nodes row by row, which do not create cycles by their addition. To decide the nodes to be added in a row, {first sort the multi-set $\{a_i,a_{i-1},...,a_1\}$ in decreasing order. Let the sorted multi-set be $\{a_{(i)},a_{(i-1)},...,a_{(1)}\}$, i.e., $a_{(i)}$ is the largest element of the multi-set and  $a_{(1)}$ is the smallest element of the multi-set. For $\overline{(5,3,2)}_2-$ICP, $a_{(j)}=a_j$, $\forall j\in [3].$} 
Now, we shuffle the columns 2 and 3 of Table \ref{Tab:icplb13e} so that User 1 contains interference nodes in rows 2 to $a_{(j)}+1$ in Column $i-j+1$ (see Table \ref{Tab:icplb3e}), i.e., the columns are rearranged so that the number of interference nodes in the first chunk for User 1 in each given column are in non-increasing order. Then, for each subsequent row, we start adding nodes from the left and continue till we encounter a side-information node.   Note that, by rearranging the columns, the ICP doesn't change.
  Now, form the acyclic induced subgraph as follows: the first set of nodes are User 1's requested files (Row 1 nodes), then we add $i=3$ left nodes from each of the $a_1=2$ rows from Row $2$ to Row $a_1+1=3$, then we add $i-1=2$ left nodes from each of the $a_2-a_1=1$ rows from Row $a_1+2=4$ to Row $a_2+1=4$ and finally, add $i-2=1$ nodes from each of the $a_3-a_2=2$ rows from Row $a_2+2=5$ to Row $a_3+1=6$. First set of nodes are Row 1 nodes $(1,1)$, $(1,2)$ and $(1,3)$. Then, we add 3 Row 2 nodes $(2,1)$, $(2,2)$ and $(2,3)$.  Then, we add 3  Row 3 nodes $(3,1)$ $(3,2)$ and $(3,3)$. Then, we add 2 Row 4 nodes $(4,1)$, and $(4,3)$. Then, we add 1 Row 5 node $(5,1)$ and finally we add 1 Row 6 node $(6,1)$. We highlight all these nodes with the green color in Table \ref{Tab:icplb3e}. Note that in Column $j$ $\forall j\in[i]$, if Node $(k,j)$ is an interference node to Row 1 nodes, and none of the nodes $(l,j)$ such that $l<k$ are side-information to Row 1 nodes then Node $(k,j)$ is also an interference node to Row $m$ nodes $\forall m<k$. Thus, all the newly added nodes in each step are interference nodes to already existing nodes and they don't have incoming edges from already added nodes. Therefore, the newly added nodes don't form cycles with the already added nodes.
Therefore, the acyclic induced sub-graph contains $K - (i-1)L + i - 1=13$ nodes $(1,1)$, $(1,2)$, $(1,3)$, $(2,1)$ $(2,2)$, $(2,3)$, $(3,1)$, $(3,2)$, $(3,3)$, $(4,1)$, $(4,2)$, $(5,1)$ and $(6,1)$ and from Lemma \ref{lemma:icplowerbound}, $R^* \ge 13$. Now, we discuss  the general $\overline{(a_i,a_{i-1},...,a_1)}_L-$ICP lower bound proof.  

\begin{proof} {\bf \hspace{-0.15in}Proof of Theorem \ref{thm:icplb}:}
 An $\overline{(a_i,a_{i-1},...,a_1)}_L-$ICP can be represented by a  $K\times i $ table. The $\overline{(5,3,2)}_2-$ICP above was represented in Table \ref{Tab:icplb13e}. In this tabular representation, Column $p$ represents the $(a_{<p-1>_i},a_{<p-2>_i},...,a_{p})_L-$ICP. For User 1, in  Column $p$ of the table, the number of interference nodes in  the first chunk from the top are $a_{<p-1>_i}$, i.e., in  Column $p$ of the table, the nodes in rows 2 to $a_{<p-1>_i}$ are interference nodes. In the $\overline{(5,3,2)}_2-$ICP, columns 1, 2 and 3 of Table \ref{Tab:icplb13e}  represent the ${(5,3,2)}_2-$ICP, ${(2,5,3)}_2-$ICP and ${(3,2,5)}_2-$ICP respectively. In Table \ref{Tab:icplb13e}, for User 1, the number of interference nodes in  the first chunk from the top in columns 1, 2 and 3 are $a_{3}=5$, $a_1=2$ and $a_2=3$ respectively. First sort the multi-set $\{a_i,a_{i-1},...,a_1\}$ in decreasing order. Let the sorted multi-set be $\{a_{(i)},a_{(i-1)},...,a_{(1)}\}$, i.e., $a_{(i)}$ is the largest element of the multi-set and  $a_{(1)}$ is the smallest element of the multi-set. Now, we rearrange the columns such that User 1 contains interference nodes in rows 2 to $a_{(j)}+1$ in Column $i-j+1$, i.e., the columns are rearranged so that the number of interference nodes in the first chunk for User 1 in each given column are in non-increasing order.  Note that, by rearranging the columns, the ICP doesn't change. Recall that Table \ref{Tab:icplb3e} represents the rearranged table of Table \ref{Tab:icplb13e}. In Table \ref{Tab:icplb3e}, for User 1, the number of interference nodes in  the first chunk from the top in columns 1, 2 and 3 are $5$, $3$ and $2$ respectively, i.e., the columns are rearranged so that the number of interference nodes in the first chunk for User 1 in each given column are in non-increasing order.

Similar to the $\overline{(5,3,2)}_2-$ICP, an $\overline{(a_i,a_{i-1},...,a_1)}_L-$ICP can be represented as an SUICP with $i\cdot K$ virtual users, each one requesting  a distinct file.  Recall that an SUICP with $i\cdot K$ virtual users can equivalently be represented by a side information graph $\mathcal{G}$ with $i\cdot K$ nodes such that each node represents a unique user and there exists an edge from Node $k$ to Node $j$ if  User $j$'s requested file belongs to  \textit{Known-set} $\mathcal{K}_k$ of User $k$. The number of nodes in any  acyclic induced sub-graph of the side-information graph $\mathcal{G}$ is  a lower bound on the broadcast rate of the corresponding ICP.


For a general $\overline{(a_i,a_{i-1},...,a_1)}_L-$ICP, we construct an acyclic induced sub-graph for our ICP from the re-arranged table as follows: First, we take all the $i$ nodes corresponding to User 1. We then continue to add all $i$ nodes from Row 2 to Row $a_{(1)} + 1$. From the definition of $a_{(1)}$, we do not encounter any side-information nodes in this step. Next, for each of the $a_{(2)} - a_{(1)}$ rows from Row $a_{(1)} + 2$ to Row $a_{(2)}+1$, we add $(i-1)$ nodes starting from the left. We continue this 
 until we add 1 left node from each of the $a_{(i)}-a_{(i-1)}$ rows from Row $a_{(i-1)}+2$ to Row $a_{(i)}+1$.
Note that, the added nodes are User 1's nodes and all its interference nodes in the first chunk from the top across the $i$ columns. This further implies that in each step, the added nodes are interference nodes for all the nodes added previously. 
Therefore, while adding the new nodes, they don't have edges  from the already added nodes and thus don't form any cycles. 

The total number of nodes in the above acyclic induced sub-graph are  $i(a_{(1)}+1)+(i-1)(a_{(2)}-a_{(1)})+(i-2)(a_{(3)}-a_{(2)})+...+2(a_{(i-1)}-a_{(i-2)})+(a_{(i)}-a_{(i-1)})$, which is equal to $i+a_{(1)}+a_{(2)}+...+a_{(i)}= i+\sum_{j=1}^{i}a_j=i+K-(i-1)L-1.$
From Lemma \ref{lemma:icplowerbound},
$ R^*\geq K-(i-1)L+i-1.$ 
\end{proof}
\subsection{Proof of Theorem \ref{thm:exact2}}\label{sec:exactproof}
\begin{proof}
	Recall that, in an $\overline{(a_2,a_1)}_L-$ICP, each user wants $2$ files and has $2L$ files as side-information. To understand the structural properties of the ICP,  we form a $K\times 2$ table such that the entry in the $k^{th}$ row and $j^{th}$ column corresponds to the $j^{th}$ file requested by User $k$, see Table \ref{Tab:exact}.
	\begin{table}[h]
		\centering
		\begin{tabular}{| c | c |}
			\hline 
			$x_{1,1}$ &  $x_{1,2}$ \\
			\hline 
			$x_{2,1}$ &  $x_{2,2}$ \\
			\hline 
			\vdots & \vdots \\
			\hline
			$x_{K,1}$ &  $x_{K,2}$ \\
			\hline 
		\end{tabular}
		\vspace{0.2cm}
		\caption{\sl $\overline{(a_2,a_1)}_L-$ICP.  } \label{Tab:exact}
		\vspace{-0.75cm}
	\end{table}
	Note that the $\overline{(a_2,a_1)}_L-$ICP is a UICP with $K$ users and each user requesting $2$ distinct files. We convert this UICP into an SUICP with $2 K$ virtual users, each requesting  a distinct file. In particular, each user in the UICP is  mapped to $2$ virtual users in the SUICP, with each virtual user requesting a distinct file from amongst the original user's requested files.  The side-information at the virtual user is the same as 
	its corresponding original user. 
	
	\textit{A proper coloring scheme:}	We take $a_1+a_2+2$ colors, and for Node ($k,1$), assign Color $<k>_{a_1+a_2+2}$ and for Node ($k,2$), assign Color $<k+a_2+1>_{a_1+a_2+2}$, $\forall k\in[K]$. Note that every color occurs in a column exactly $K/(a_1+a_2+2)$ times. Now, we verify that it is a proper coloring scheme. 
	
	Consider a node from Column 1, say Node ($k,1$). Note that its color is $<k>_{a_1+a_2+2}$ and its interfering nodes in Column 1 are Nodes $(<k-j>_{K},1)$ $\forall j\in[a_1]$ and Nodes $(<k+j>_{K},1)$ $\forall j\in[a_2]$ and their colors are $<k-j>_{a_1+a_2+2}$ $\forall j\in[a_1]$ and  $<k+j>_{a_1+a_2+2}$ $\forall j\in[a_2]$ respectively. We can easily verify that none of these colors overlap with Color $<k>_{a_1+a_2+2}$. Hence, we can say that Node ($k,1$) doesn't share its color with its  first column interference nodes. Node ($k,1$)'s interfering nodes in Column 2 are Nodes $(<k-j>_{K},2)$ $\forall j\in[a_2]$, Node ($k,2$), and Nodes $(<k+j>_{K},2)$ $\forall j\in[a_1]$ and their colors are $<k-j+a_2+1>_{a_1+a_2+2}$ $\forall j\in[a_2]$, $<k+a_2+1>_{a_1+a_2+2}$ and  $<k+j+a_2+1>_{a_1+a_2+2}$ $\forall j\in[a_1]$ respectively. Again, we can easily verify that none of these colors are Color $<k>_{a_1+a_2+2}$. Hence, we can say that Node ($k,1$) doesn't share its color with its  second column interference nodes either. Hence, we can say that under the proposed coloring scheme, none of the nodes of Column 1 share their colors with any of the corresponding interfering nodes. By using similar arguments, we can show that the same holds true for all nodes in Column 2 as well. Therefore, the coloring scheme is proper.
	
	
	Since, we are using only  $a_1+a_2+2$ colors, every node contains at most $a_1+a_2+2$ colors in its closed anti-outneighborhood. Therefore, from Lemma  \ref{lemma:icpupperbound}, we have $R^*\leq\mathcal{X}_l\leq a_1+a_2+2$. On the other hand, we have from Theorem \ref{thm:icplb} that
	$R^*\geq \sum_{j=1}^{2}(a_j+1)=a_1+a_2+2.$
\end{proof}
\subsection{Proof of Theorem \ref{thm:multiub2}}\label{sec:maccproof}
First, we define an $\widetilde{(\mathbf{s}\times l)}_L-$ICP with $\mathbf{s}=(a_j,a_{j-1},...,a_1)^T$, and characterize its transmission rate.  These results are useful in the proof of Theorem \ref{thm:multiub2}. In this section, we represent the $(a_i,a_{i-1},...,a_1)_L-$ICP and $\overline{(a_i,a_{i-1},...,a_1)}_L-$ICP   as $(\mathbf{a})_L-$ ICP and $\overline{(\mathbf{a}})_L-$ICP respectively, where $\mathbf{a}$ is a vector of the dimension $i$ and is given by $\mathbf{a}=(a_i,a_{i-1},...,a_1)^T$.

Let $\mathbf{s}=(s_m,s_{m-1},...,s_1)^T$ be a vector of dimension $m$ and $l$ be an integer, then we define $\mathbf{s}\times l$ as a vector of dimension $l\cdot m$ given by $\mathbf{s}\times l=(s_m,...,s_1,s_m,...,s_1,..., l \text{ times }...,s_m,...,s_1)^T$. For example, if $\mathbf{s}=(a_2,a_1)^T$ then $\mathbf{s}\times 3=(a_2,a_1,a_2,a_1,a_2,a_1)^T$.

Let $\mathbf{s}=(a_m,a_{m-1},...,a_1)^T$ be a vector such that $a_m\geq a_k$, $\forall k\neq m$ and $\mathbf{s}_t$ be the vector formed by the $t$ clockwise rotations of the vector $\mathbf{s}$. Then, we define  $\widetilde{(\mathbf{s}\times l)}_L-$ICP as the union of  $(\mathbf{s}\times l)_L-$ICP, $(\mathbf{s}_1\times l)_L-$ICP,...,$(\mathbf{s}_{j-1}\times l)_L-$ICP.
\begin{example}
	Consider a vector $\mathbf{s}=(2,0,2)^T$ and $l=2$. Then $\mathbf{s}_1=(2,2,0)^T$, $\mathbf{s}_2=(0,2,2)^T$ and  the $\widetilde{(\mathbf{s}\times 2)}_L-$ICP is the union of   $(2,0,2,2,0,2)_L-$ICP, $(2,2,0,2,2,0)_L-$ICP, and  $(0,2,2,0,2,2)_L-$ICP. On the other hand, the $\overline{(\mathbf{s}\times 2)}_L-$ICP is the union of $(2,0,2,2,0,2)_L-$ICP, $(2,2,0,2,2,0)_L-$ICP,  $(0,2,2,0,2,2)_L-$ICP, $(2,0,2,2,0,2)_L-$ICP, $(2,2,0,2,2,0)_L -$ICP, and  $(0,2,2,0,2,2)_L-$ICP. In general,  the $\overline{(\mathbf{s}\times l)}_L-$ICP is a union of $l$ copies of each individual ICP in a $\widetilde{(\mathbf{s}\times l)}_L-$ICP. 
\end{example}
\begin{lemma}\label{partialub}
	Let $R^*$ be the optimal broadcast rate of the $\widetilde{(\mathbf{s}\times l)}_L-$ICP with $\mathbf{s}=(a_m,a_{m-1},...,a_1)^T$. Then 
	\begin{align*}
	\sum_{t=1}^{m}(a_t+1)\leq	R^*\leq	
		l\min\{2(K-(i-1)L)+i-2-a_j,K\},	
	\end{align*}
	where $i=l\cdot m$ and $K=l(\sum_{k=1}^{m}a_k)+(i-1)L+1$ is the number of users in the $\widetilde{(\mathbf{s}\times l)}_L-$ICP.
\end{lemma}
The proof of the above lemma is given in Appendix \ref{sec:partialub}. Now, we discuss the proof of Theorem \ref{thm:multiub2}.

}
%
\begin{proof}{\bf Proof of Theorem \ref{thm:multiub2}:}
	We provide a proof sketch before giving the details. In the placement phase of the CCDN, we use the same uncoded placement policy as the one proposed in \cite{reddy2020rate} and in the delivery phase, we form an instance of the ICP with the required subfiles and use the results mentioned in Section \ref{ICP_results} to derive an upper bound on the server transmission rate. In particular, we first split our ICP into many  ICP's of the form $\overline{(a_i,a_{i-1},...,a_1)}-$ICP. Then, we use Theorem \ref{thm:icpub} to get an upper bound on the data transmission rate of each $\overline{(a_i,a_{i-1},...,a_1)}-$ICPs. Finally, the data transmission rate of multi-access ICP is upper bounded by the sum of the upper bounds for the individual ICPs. The details are as follows.
	
Recall from Section \ref{sec:macc_setting} that in the placement phase, we divide each file into $|\hat{\mathcal{S}}|$ parts of equal size, one corresponding to each subset $s\in\hat{\mathcal{S}}$. Then, we store the sub-file assigned to the set $s$, in all the $w=MN/K$ caches whose index belongs to $s$. We can observe that
\begin{itemize}
	\item if a sub-file is assigned to the subset $\{<a_1>_K,$ $<a_2>_K, ..., <a_w>_K\}$, then it is available to the $wL$ users $[<a_1-L+1>_K:<a_1>_K]$, $[<a_2-L+1>_K:<a_2>_K]$, ..., $[<a_w-L+1>_K:<a_w>_K]$. 
	From the definition of $\hat{\mathcal{S}}$, no two user sets above $[<a_j-L+1>_K:<a_j>_K]$,  $[<a_k-L+1>_K:<a_k>_K]$ have any common elements $\forall j\neq k$.
	 We represent the corresponding sub-file of File $j$ as $\mathcal{F}_{j,\{[<a_1-L+1>_K:<a_1>_K]\cup ...\cup [<a_w-L+1>_K:<a_w>_K]\}}$.   Here, the subscripted set represents that the sub-file is available at the users mentioned in the set.
	\item the storage policy is symmetrical, i.e., each cache stores equal number of subfiles for every file. This follows from the assertion that the number of subsets $s\in \hat{\mathcal{S}}$ containing any index $j\in [K]$ are the same. This is because if index 1 is present in subset $\{<a_1>_K,<a_2>_K,...,<a_w>_K\}\in\hat{\mathcal{S}}$ then index  $j\in [K]$ is present in the subset $\{<a_1+j-1>_K,<a_2+j-1>_K,...,<a_w+j-1>_K\}$, which  also belongs to $\hat{\mathcal{S}}$. 
	The number of subsets containing the index $j$  is ${K-wL+w-1 \choose w-1}$, i.e., each cache stores ${K-wL+w-1 \choose w-1}$ subfiles of each file. 
	\item from the definition of $\hat{\mathcal{S}}$, no two caches $j$, $k(\neq j)$ have a common file if $|j-k|<L$ or $|K-|j-k||<L$ which implies that for any user, all accessible caches store distinct content. 
\end{itemize}

Let the request pattern be $\{d_1,d_2,...,d_K\}$, i.e., User $i$ requests File $\mathcal{F}_{d_i}$. Among the $|\hat{\mathcal{S}}|$ subfiles of File $\mathcal{F}_{d_i}$, the user already has access to $L{K-wL+w-1 \choose w-1}$ subfiles, which are stored in its associated $L$ caches, and the user only needs to recover the remaining $|\hat{\mathcal{S}}|-L{K-wL+w-1 \choose w-1}={K-wL+w-1 \choose w}$ subfiles. 
Thus a   total of $K{K-wL+w-1 \choose w}$ subfiles{\footnote{The total number of requests are calculated by assuming that all the requests are distinct. The case where some requests are repeated can be handled similarly.}}  are needed across the $K$ users.

We map the problem here to an instance of the ICP with $K{K-wL+w-1 \choose w}$ virtual users/nodes such that each one requests a distinct subfile. The side-information at a virtual user is the same as the subfiles available to the real user requesting the corresponding subfile. To understand the structural properties of the ICP, we form a $K\times{K-wL+w-1 \choose w}$ table such that
\begin{enumerate}
	\item each cell represents a virtual user,
	\item $l^{th}$ row represents User $l$'s required subfiles,
	\item if a column's $1^{st}$ element is $\mathcal{F}_{d_1,\{\cup_{j=1}^w[<a_j-L+1>_K:<a_j>_K]\}}$ then for all  $k\in[K]$, its  $k^{th}$ element is $\mathcal{F}_{d_k,\{\cup_{j=1}^w[<a_j-L+k>_K:<a_j+k-1>_K]\}}.$
\end{enumerate}
See Table \ref{ICPex} for an illustration. 

Since all the cells in a row are requested by the same user,    for all the cells in a row the side-information cells are same. In particular,  all the cells which contain $l$ in the subscript are available at user $l$ and hence are side-information cells to the cells in Row $l$. Property 3 above implies that the relative positions of available subfiles are the same for all the virtual users in a column and hence each column alone represents a symmetric ICP \cite{maleki2014index}. Note that (\emph{i}) the subscripted set of each subfile is the disjoint union of $w$ subsets of $[K]$, each with $L$ consecutive elements and (\emph{ii}) in each column,  the subscripts are circularly shifting by one from Row $j$ to Row $j+1$. Thus for any column, each element $k \in [K]$ occurs in $wL$ entries, partitioned into $w$ chunks each of size $L$. Hence, the side-information for any node in a column also follows the same pattern. 
Therefore, each column represents a $(b_{w+1},b_w,...,b_1)_L-$ICP for some $b_j$s and $\sum_{j=1}^{w+1}b_j=K-wL-1$. Here, $b_j$s are the gaps between the chunks of side-information. 

Recall from Section \ref{sec:macc_setting} that a vector $\mathbf{b }= (b_{m},b_{m-1},...,b_{1})^T$ of dimension $m$ is said to be a weak $m$ composition of $n$ \cite{stanley2011enumerative} if the  components  of $\mathbf{b }$ are non-negative and their sum is $n$. Let $\mathcal{B}$ denote the collection of all weak $w+1$ compositions  of $K-wL-1$. 
 Note that $|\mathcal{B}|={K-wL+w-1 \choose w}$. Also recall from the definition of $\hat{\mathcal{S}}$ in Section \ref{sec:macc_setting} that $\hat{\mathcal{S}}$ is the collection of all possible subsets $s$ of $[K]$ satisfying the constraints (\emph{i}) $|s|=w$, (\emph{ii})  every two different elements $a_j, a_l$ of $s$ satisfy $|a_j-a_l|\geq L$ and $|K-|a_j-a_l|| \geq L$. It can be shown that these properties imply that for each vector $\mathbf{b}=[b_{w+1},b_w,...,b_1]^T$ in $\mathcal{B}$, there exists a column in the $K\times {K-wL+w-1 \choose w}$ table such that the cells in the column form a  $(b_{w+1},b_w,...,b_1)_L-$ICP and vice versa. 
   Recall that in Example \ref{ex:multiaccess}, $\mathcal{B}=\{({3,0,0})^T,(0,3,0)^T,(0,0,3)^T,({2,1,0})^T,(0,2,1)^T, (1,0,2)^T,$ $({2,0,1})^T,(1,2,0)^T,(0,1,2)^T,({1,1,1})^T\}$. For each element in $\mathcal{B}$, there exists a  column in Table \ref{siex} with its corresponding ICP and vice versa. 
For example in Table \ref{siex}, columns 5 and 6 form the $(2,1,0)_2-$ICP and the $(1,1,1)_2-$ ICP respectively.

Let $\mathcal{C}$ be the smallest subset of $\mathcal{B}$ such that the vectors $\mathbf{c }=$ $(c_{w+1},c_w,...,c_{1})^T$ in $\mathcal{C}$ are of the form  $0\leq c_j\leq c_{w+1}$ $\forall j\in[w+1]$, and the vectors and their possible clockwise rotations (like $(c_{2},c_1,c_{w+1},...,c_{3})^T$ ) cover all the vectors in $\mathcal{B}$. Consider a vector $\mathbf{c }_j\in \mathcal{C}$. Then, using the notation introduced at the beginning of this section, $\mathbf{c }_j$ can be written as $\mathbf{u}_j\times v$,  where $v$ is a positive integer and $\mathbf{u}_j$ is a vector such that its $len(\mathbf{u}_j)$ clockwise rotations are distinct. Here $len(\mathbf{u}_j)$ denotes the number of elements in $\mathbf{u}_j$.  In Example \ref{ex:multiaccess}, $\mathcal{C}=\{({3,0,0})^T,({2,1,0})^T,({2,0,1})^T,({1,1,1})^T\}$ and the element $({2,0,1})^T$ can be written as $({2,0,1})^T\times 1$, while the element $({1,1,1})^T$ can be written as $(1)^T\times 3$.

 For any $\mathbf{c}_j$ in $\mathcal{C}$, let $\mathcal{B}_j$ be the collection of $\mathbf{c}_j$  $(=\mathbf{u}_j\times v)$ and its clockwise rotations. Note that $|\mathcal{B}_j|=len(\mathbf{u}_j)$, i.e., the cardinality of the set $\mathcal{B}_j$ is equal to the number of elements in the vector  $\mathbf{u}_j$. Also note that $len(\mathbf{c}_j)= |\mathcal{B}_j|\cdot v =w+1$. In Example \ref{ex:multiaccess}, the $\mathcal{B}_j$'s corresponding to the elements $({2,0,1})^T\times 1=({2,0,1})^T$   and $(1)^T\times 3=({1,1,1})^T$ are $\{({2,0,1})^T,({1,2,0})^T,({0,1,2})^T\}$ and $\{({1,1,1})^T\}$ respectively. Consider the $|\mathcal{B}_j|$ columns represented by the elements of $\mathcal{B}_j$ in the $K\times{K-wL+w-1 \choose w}$ table. Again recalling the notation defined in the beginning of the section, note that these columns together form the $\widetilde{(\mathbf{u}_j\times v)}_L-$ICP. In Example \ref{ex:multiaccess}, the $\mathcal{B}_j$ corresponding to the element  $({1,1,1})^T$ is $\{({1,1,1})^T\}$ and the corresponding column 6 in Table \ref{siex} forms the $(\widetilde{(1)^T\times 3})_2-$ICP.  Using Lemma \ref{partialub}, the upper bound on the data transmission rate for the $\widetilde{(\mathbf{u}_j\times v)}_L-$ICP  is $\frac{len(\mathbf{u}_j)\cdot \min\{2(K-wL)+w-1-\mathbf{\widehat{u}}_j,K\}}{w+1}$ which is equal to  $\sum_{\mathbf{b}\in \mathcal{B}_j}\frac{\min\{2(K-wL)+w-1-\mathbf{\widehat{b}},K\}}{w+1}$. In general, the above calculations are true for any $j\in[|\mathcal{C}|]$.

According to the definition of $\mathcal{C}$, the vectors in $\mathcal{C}$ and their clock wise rotations cover all the columns in the $K\times{K-wL+w-1 \choose w}$ table. Therefore, the data transmission rate of the ICP formed by the $K\times{K-wL+w-1 \choose w}$ table is upper bounded by $\sum_{j=1}^{|\mathcal{C}|}\sum_{\mathbf{b}\in \mathcal{B}_j}\frac{\min\{2(K-wL)+w-1-\mathbf{\widehat{b}},K\}}{w+1}$. Since $\mathcal{B}$ is the disjoint union of $\mathcal{B}_j$ $\forall j\in[|\mathcal{C}|]$, $\sum_{j=1}^{|\mathcal{C}|}\sum_{\mathbf{b}\in \mathcal{B}_j}\frac{\min\{2(K-wL)+w-1-\mathbf{\widehat{b}},K\}}{w+1}$ is equal to $\sum_{\mathbf{b}\in \mathcal{B}}\frac{\min\{2(K-wL)+w-1-\mathbf{\widehat{b}},K\}}{w+1}$.
 

  We now use the multi-access ICP upper bound derived above to give an upper bound on the server transmission rate of the ($N, K, L$)$-$CCDN at memory point $M=wN/K$. Till now, we calculated the upper bound on the ICP rate assuming a unit subfile size. However, in the multi-access ICP, each subfile is of size $1/|\hat{\mathcal{S}}|$ units. 
  Therefore, for ($N, K, L$)$-$CCDN at memory point $M=wN/K$, $$R_\text{new}(wN/K)\leq \frac{1}{|\hat{\mathcal{S}}|} \sum_{\mathbf{b}\in \mathcal{B}}\frac{\min\{2(K-wL)+w-1-\mathbf{\widehat{b}},K\}}{w+1}$$ units, where $\mathcal{B}$ is the collection of weak $w+1$ compositions of $K-wL-1$.
\end{proof}

\subsection{Proof of Corollary \ref{cor:exactspecial} }\label{sec:corollary4proof}
\begin{proof}
	From \eqref{eq:sumcondn}, for an $\overline{(a_i,0,0,...,0)}_L-$ICP, $a_i=K-(i-1)L-1.$ From Theorem \ref{thm:icpub},
	\begin{align*}
		R^*_1&\leq 2(K-(i-1)L)+i-2-a_i\\ &=2(K-(i-1)L)+i-2-(K-(i-1)L-1)\\&=K-(i-1)L+i-1.
	\end{align*}
From Theorem \ref{thm:icplb}, 
\begin{align*}
	R^*_1&\geq K-(i-1)L+i-1.
\end{align*}
Therefore, for an $\overline{(a_i,0,0,...,0)}_L-$ICP,
$$R^*_1= K-(i-1)L+i-1.$$
From Theorem \ref{thm:icpub}, for an $\overline{(a_i,a_{i-1},...,a_1)}_1-$ICP,
\begin{align*}
	R^*_2\leq K.
\end{align*}
From Theorem \ref{thm:icplb}, 
\begin{align*}
	R^*_2\geq K-(i-1)L+i-1=K-(i-1)\times 1+i-1=K.
\end{align*}

Therefore, for an $\overline{(a_i,a_{i-1},...,a_1)}_1-$ICP,
$$R^*_2= K.$$
\end{proof}

\subsection{Multi-access coded caching problem achivable rates at the memory points $M=0$ and $M=\lceil K/L \rceil \cdot N/K$}\label{sec:multiaccproof}
\begin{lemma}
	For an $(N,K,L)-$CCDN, at the memory points $M=0$ and $M=\lceil K/L \rceil \cdot N/K$, we achieve the rates $K$ units and $0$ units respectively.
\end{lemma}
\begin{proof}\textbf{Proof:}
\subsubsection{$M=0$}
At $M=0$, no file parts are stored in the cache. In the worst-case, all users request $K$ different files. Hence, $R^*(0)\leq K\text{ units.}$
\subsubsection{$M=\lceil K/L \rceil \cdot N/K$}
We first create a list (of size $\big\lceil\frac{K}{L}\big\rceil\times N$ elements) by repeating the sequence $\{1,2,...,N\}$ $MK/N=\lceil{K}/{L}\rceil$ times, i.e., $\{1,2,...N-1,N,1,2,...N-1,N,1,2,...,\}$. Then we fill the caches according to the list sequentially. Note that  the total memory required to fit the list is $\lceil{K}/{L}\rceil\cdot N$ units, which is equal to our total cache memory. Hence, the memory constraint is satisfied. 
This storage policy makes sure that  each user has access to all files in the central server's catalog.  Hence, the worst-case transmission rate is $R^*(\lceil K/L \rceil \cdot N/K)=0$ units.	

\end{proof}

\subsection{Proof of Lemma \ref{lemma:interference}}\label{sec:interference}
\begin{proof}
	Recall \eqref{eq:side-info2} from Section \ref{sec:icp_setting} that 
	\begin{align}
		\mathcal{K}^n_k=\{&x_{b,n}: v\in[i-1], r\in[L], b=<k+\sum_{l=1}^{v}a_{<n-l>_i}+(v-1)L+r>_K\},\label{eq:km}\\
		\mathcal{K}_k^{n+1}=&\{x_{b,n+1}:v\in[i-1], r\in[L],
		b=<k+\sum_{l=1}^{v}a_{<n+1-l>_i}+(v-1)L+r>_K \}.\label{eq:km+1}
	\end{align} 
	First we prove the side-information node case. If Node $(<k+m>_K,n+1)$ is a side-information node, then there exists $v'\in[2:i-1], r'\in[L]$ in \eqref{eq:km+1} such that $<k+\sum_{l=1}^{v'}a_{<n+1-l>_i}+(v'-1)L+r'>_K=<k+m>_K$. Node $(<k+m-a_n-L>_K,n)$ is equal to Node $(<k+\sum_{l=1}^{v'}a_{<n+1-l>_i}+(v'-1)L+r'-a_n-L>_K,n)$, and it belongs to $\mathcal{K}^n_k$ with $v=v'-1, r=r'$ in \eqref{eq:km}. Therefore, if Node $(<k+m>_K,n+1)$ is a side-information node then Node $(<k+m-a_n-L>_K,n)$ is also a side-information node.
	
	Next, we prove the interference  node case by contradiction. If Node $(<k+m>_K,n+1)$ is an interference node and node $(<k+m-a_n-L>_K,n)$ isn't an interference node, then  there exists $v'\in[i-2], r'\in[L]$ in \eqref{eq:km} such that $<k+\sum_{l=1}^{v'}a_{<n-l>_i}+(v'-1)L+r'>_K=<k+m-a_n-L>_K$. Node $(<k+m>_K,n+1)$ is equal to Node $(<k+\sum_{l=1}^{v'}a_{<n-l>_i}+(v'-1)L+r'+a_n+L>_K,n+1)$, and it belongs to $\mathcal{K}^{n+1}_k$ with $v=v'+1, r=r'$ in \eqref{eq:km+1}. It contradicts the fact that Node $(<k+m>_K,n+1)$ is an interference node. Therefore, if Node $(<k+m>_K,n+1)$ is an interference node then Node $(<k+m-a_n-L>_K,n)$ is also an interference node.
\end{proof}
\subsection{Proof of Lemma \ref{lemma:proper}}\label{sec:proper}
\begin{lemma}\label{lemma:elementary}
	Some elementary observations about our proposed coloring scheme in the proof of Lemma \ref{lemma:proper} are as follows:
	\begin{enumerate}
		\item According to our coloring scheme, if Color $c_1$ is assigned to Node ($k,j$) then $$<c_1+\sum_{m=1}^{j-1}a_{m}+(j-1)L>_K=k.$$
		\item If $<c_1+\sum_{m=1}^{j-1}a_{m}+(j-1)L>_K=k$, and $p>j$, then
		\begin{align*}
			<c_1+\sum_{m=1}^{p-1}a_{m}+(p-1)L>_K
			=<k+\sum_{m=1}^{p-j}a_{<p-m>_i}+(p-j)L>_K.
		\end{align*} 
		\item If $<c_1+\sum_{m=1}^{j-1}a_{m}+(j-1)L>_K=k$, and $p<j$, then
		\begin{align*}
			<c_1+\sum_{m=1}^{p-1}a_{m}+(p-1)L>_K=<k+\sum_{m=1}^{i-j+p}a_{<p-m>_i}+(i-j+p-1)L+1>_K.
		\end{align*} 
	\end{enumerate}
\end{lemma}
\begin{proof}\textbf{Proof:}
	\begin{enumerate}
		\item According to our coloring scheme, Color $c$ is assigned to Node $(<c+\sum_{m=1}^{v-1}a_{m}+(v-1)L>_K,v)$ $\forall v\in[i]$ and every color occurs in a column exactly once. Therefore, if Color $c_1$ is assigned to Node ($k,j$) then $<c_1+\sum_{m=1}^{j-1}a_{m}+(j-1)L>_K=k.$
		\item \begin{align*}
			<c_1+\sum_{m=1}^{p-1}a_{m}+(p-1)L>_K
			=&<c_1+\sum_{m=1}^{j-1}a_{m}+(j-1)L+\sum_{m=j}^{p-1}a_{m}+(p-j)L>_K\\
			=&<k+\sum_{m=j}^{p-1}a_{m}+(p-j)L>_K\\
			=&<k+\sum_{m=1}^{p-j}a_{<p-m>_i}+(p-j)L>_K.
		\end{align*}
		\item 
	\begin{align*}
		<c_1+\sum_{m=1}^{p-1}a_{m}+(p-1)L>_K
		=&<c_1+\sum_{m=1}^{p-1}a_{m}+(p-1)L+\sum_{m=1}^{i}a_{m}+(i-1)L+1-K>_K\\
		=&<k+\sum_{m=1}^{p-1}a_{m}+(p-j)L+\sum_{m=j}^{i}a_{m}+(i-1)L+1>_K\\
		=&<k+\sum_{m=1}^{i-j+p}a_{<p-m>_i}+(i-j+p-1)L+1>_K.
	\end{align*} .
	\end{enumerate}
\end{proof}

Now, we prove Lemma \ref{lemma:interference}.

\begin{proof}{\bf \hspace{-0.15in}Proof of Lemma \ref{lemma:interference}:}
	Recall \eqref{eq:side-info} from Section \ref{sec:icp_setting} that
	\begin{align}\label{eqn:kk}
		\mathcal{K}_k=\{x_{b,t}:t\in[i], v\in[i-1], r\in[L],b=<k+\sum_{m=1}^{v}a_{<t-m>_i}+(v-1)L+r>_K
		\}.
	\end{align}
	Recall from Section \ref{sec:ICP_preliaries} that for an SUICP, a proper coloring scheme assigns a color to each user such that no user shares its color with any of its interfering users.
		We take $K$ colors and assign Color $c$ to nodes $(<c+\sum_{m=1}^{v-1}a_{m}+(v-1)L>_K,v)$ $\forall v\in[i]$. Note that, every color occurs in a column exactly once. 
	 Now, we verify that our coloring scheme is a proper coloring scheme, i.e., for any Node ($k,j$), its color is assigned only to other nodes in its side-information set. If Color $c_1$ is assigned to Node ($k,j$), then from Lemma \ref{lemma:elementary} (1), $<c_1+\sum_{m=1}^{j-1}a_{m}+(j-1)L>_K=k$. For all $p\in[i]$, Color $c_1$ is assigned to Node $(<c_1+\sum_{m=1}^{p-1}a_{m}+(p-1)L>_K,p)$ in $p^{th}$ Column. 
	 \begin{itemize}
	 	\item If $p>j$, from Lemma \ref{lemma:elementary} (2), Node $(<c_1+\sum_{m=1}^{p-1}a_{m}+(p-1)L>_K,p)$ is equal to Node $(<k+\sum_{m=1}^{p-j}a_{<p-m>_i}+(p-j)L>_K,p)$ and we can verify that it belongs to $\mathcal{K}_k$,  by substituting $t=p,  r=L, v=p-j$ in \eqref{eqn:kk}. Hence, the coloring scheme ensures that the color of  Node ($k,j$)  is repeated in Column $p$ only at a side-information node of Node ($k,j$) $\forall p>j$.
	 	\item If $p<j$, from Lemma \ref{lemma:elementary} (3), Node $(<c_1+\sum_{m=1}^{p-1}a_{m}+(p-1)L>_K,p)$ is equal to Node \remove{$(<m-\sum_{j=p}^{n-1}a_{j}-(n-p)L>_K,p)=(<m-\sum_{j=p}^{n-1}a_{j}-(n-p)L+\sum_{j=1}^{i}a_j+(i-1)L+1>_K,p)=$}$(<k+\sum_{m=1}^{i-j+p}a_{<p-m>_i}+(i-j+p-1)L+1>_K,p)$, and we can verify that it belongs to $\mathcal{K}_k$ by substituting $t=p, r=1, v=i-j+p$ in \eqref{eqn:kk}. Hence, the coloring scheme ensures that the color of  Node ($k,j$)  is repeated in Column $p$ only at a side-information node of Node ($k,j$) $\forall p<j$.
	 	\item In a column, a color occurs only once. Hence, none of the other nodes in Column $j$ are assigned with the color of Node ($k,j$).
	 \end{itemize}
	 Therefore, this coloring scheme ensures that every node shares its color only with its non-interfering nodes. In other words, none of the nodes share its color with its interfering nodes and hence is a proper coloring scheme.
\end{proof}

\subsection{Proof of Lemma \ref{partialub}}\label{sec:partialub}
\begin{proof}
		Note that if we split a file in the $(a_i,a_{i-1},...,a_1)_L-$ICP into $l$ equal sized subfiles, then with respect to the subfiles, we have the union of $l$ $(a_i,a_{i-1},...,a_1)_L-$ICPs. Furthermore, recall that the $\overline{(\mathbf{s}\times l)}_L-$ICP is a union of $l$ copies of each individual ICP in a $\widetilde{(\mathbf{s}\times l)}_L-$ICP. Thus, if we start with the $\widetilde{(\mathbf{s}\times l)}_L-$ICP and then split each file into $l$ equal sized subfiles, we will have the $\overline{(\mathbf{s}\times l)}_L-$ICP with respect to the subfiles. Recall from the $\overline{(\mathbf{a})}_L-$ICP (=$\overline{(a_i,a_{i-1},...,a_1)}_L-$ICP) representation above that $i$ is the dimension of the vector $\mathbf{a}$ and from Section \ref{sec:icp_setting} that the number of users in the  $\overline{(\mathbf{a}})_L-$ICP is $K=\sum_{k=1}^{i}a_k+(i-1)L+1$. Therefore, for the $\overline{(\mathbf{s}\times l)}_L-$ICP, the dimension of the vector $\mathbf{s}\times l$ is $i=l\cdot m$ and the number of users of the $\overline{(\mathbf{s}\times l)}_L-$ICP is $K=l(\sum_{k=1}^{m}a_k)+(i-1)L+1$. 
		From Theorem \ref{thm:icpub}, an upper bound on the data transmission rate of the $\overline{(\mathbf{s}\times l)}_L-$ICP is $\min\{2(K-(i-1)L)+i-2-a_j,K\}$. Since, each subfile is of size $1/l$ units, an upper bound on the data transmission rate of the $\widetilde{(\mathbf{s}\times l)}_L-$ICP is 
		$R^*\leq \min\{2(K-(i-1)L)+i-2-a_j,K\}/l$ units $=m\cdot \min\{2(K-(i-1)L)+i-2-a_j,K\}/i$ units.
		
		The proof of the lower bound follows along similar lines as the proof of Theorem \ref{thm:icplb} and is based on the maximum acyclic induced subgraph bound. 
		%
\end{proof}

	\subsection{Proof of Corollary \ref{cor:closed}} \label{sec:closedproof}
	\begin{proof}
		By substituting $w = 1$ in \eqref{eqn:multiub2}, we have that for an ($N,K,L\geq K/2$)$-$CCDN at memory point $M=N/K$, 
		\begin{align*}
			R_{\text{new}}(M)=\frac{\sum_{\mathbf{b }\in \mathcal{B}}(2(K-L)-\mathbf{\widehat{b}})}{2|\hat{\mathcal{S}}|}.
		\end{align*}
		where, $\mathcal{B}=\{(a_2,a_1)\in \{\mathbb{N}\cup\{0\}\}^2: a_1+a_2=K-L-1\}.$
		First, consider the case when $K-L-1$ is odd. Let $\mathcal{B}_1=\{(a_2,a_1)\in \{\mathbb{N}\cup\{0\}\}^2: a_1\in[\frac{K-L-2}{2}],a_2=K-L-1-a_1\}$ and $\mathcal{B}_2=\{(a_2,a_1)\in \{\mathbb{N}\cup\{0\}\}^2: a_2\in[\frac{K-L-2}{2}],a_1=K-L-1-a_2\}$. Note that the disjoint union of $\mathcal{B}_1$ and $\mathcal{B}_2$ gives $\mathcal{B}$ and thus,
		\begin{align*}
			R_{\text{new}}(M)&=\frac{\sum_{\mathbf{b }\in \mathcal{B}}(2(K-L)-\mathbf{\widehat{b}})}{2|\hat{\mathcal{S}}|}\\
			&=\frac{\sum_{\mathbf{b }\in \mathcal{B}_1}(2(K-L)-\mathbf{\widehat{b}})}{2|\hat{\mathcal{S}}|}+\frac{\sum_{\mathbf{b }\in \mathcal{B}_2}(2(K-L)-\mathbf{\widehat{b}})}{2|\hat{\mathcal{S}}|}\\
			&=\frac{\sum_{i=0}^{\frac{K-L-2}{2}}(K-L+1+i)}{|\hat{\mathcal{S}}|}=\frac{(K-L)(5K-5L+2)}{8K}
		\end{align*}
		
		Now, consider  the case $K-L-1$ is   even. We partition the set $\mathcal{B}$ into 3 disjoint sets $\mathcal{B}_3$, $\mathcal{B}_4$, and $\mathcal{B}_5$, where $\mathcal{B}_3=\{(a_2,a_1)\in \{\mathbb{N}\cup\{0\}\}^2: a_1\in[\frac{K-L-3}{2}],a_2=K-L-1-a_1\}$, $\mathcal{B}_4=\{(\frac{K-L-1}{2},\frac{K-L-1}{2})\}$, and $\mathcal{B}_5=\{(a_2,a_1)\in \{\mathbb{N}\cup\{0\}\}^2: a_2\in[\frac{K-L-3}{2}],a_1=K-L-1-a_2\}$. By following similar steps as the case when $K - L - 1$ is odd, we get	
		\begin{align*}
			R_{\text{new}}(M)=\frac{(K-L)(5K-5L+2)+1}{8K}
		\end{align*}	
\end{proof}

\subsection{An upper bound on the transmission rate of the ${(1,1,1)}_2-$ICP}\label{sec:111_2ub}
\begin{lemma}\label{lem:111_2ub}
	Let $R^*$ be the optimal transmission rates for the ${(1,1,1)}_2-$ICP  Then, $$R^*\leq 8/3 \text{ units}.$$ 
\end{lemma}
\begin{proof}\textbf{Proof:}
	Recall from the definition of the $\overline{(a_i,a_{i-1},...,a_1)}_L-$ICP, in the $\overline{(1,1,1)}_2-$ICP, $\forall k\in [K]$
	\begin{itemize}
		\item  want set $\mathcal{W}_k=\{y_{k,1}, y_{k,2}, y_{k,3}\}$ and
		\item  known set \begin{align*}
			\mathcal{K}_k=\{y_{b,t}:t\in[3], v\in[2], r\in[2],b=<k+v+(v-1)2+r>_8
			\},
		\end{align*}
	i.e.,
	\begin{align}\label{eqn:111bar}
		\mathcal{K}_k=\{y_{<k+2>_8,t}, y_{<k+3>_8,t}, y_{<k+5>_8,t}, y_{<k+6>_8,t}: t\in[3]
		\}.
	\end{align}
	\end{itemize}
The tabular representation of  $\overline{(1,1,1)}_2-$ICP is given in Table \ref{Tab:111_2}.
\begin{table}[h]
		\centering
		\begin{tabular}{| c | c | c |}
			\hline 
			\color{red}${y_{1,1}}$ &  \color{red}${y_{1,2}}$ & \color{red} ${y_{1,3}}$\\
			\hline 
			$y_{2,1}$ &  $y_{2,2}$ & $y_{2,3}$\\
			\hline
			\cellcolor{gray}\color{blue}$\mathbf{y_{3,1}}$ &  \cellcolor{gray}\color{blue}$\mathbf{y_{3,2}}$ & \cellcolor{gray}\color{blue}$\mathbf{y_{3,3}}$\\
			\hline
			\cellcolor{gray}\color{blue}$y_{4,1}$ &  \cellcolor{gray}\color{blue}$y_{4,2}$ & \cellcolor{gray}\color{blue}$y_{4,3}$\\
			\hline 
			$y_{5,1}$ &  $y_{5,2}$ & $y_{5,3}$\\
			\hline
			\cellcolor{gray}\color{blue}$y_{6,1}$ & \cellcolor{gray}\color{blue} $y_{6,2}$ & \cellcolor{gray}\color{blue}$y_{6,3}$\\
			\hline
			\cellcolor{gray}\color{blue}	$y_{7,1}$ & \cellcolor{gray}\color{blue} $y_{7,2}$ & \cellcolor{gray}\color{blue}$y_{7,3}$\\
			\hline
			$y_{8,1}$ &  $y_{8,2}$ & $y_{8,3}$\\
			\hline
		\end{tabular}
		\vspace{0.2cm}
		\caption{\sl This table corresponds to  the $\overline{(1,1,1)}_2-$ICP. Here, we highlight the requested files of User 1 with the red color fonts and the side-information files of User 1 with the shaded cells and the blue color fonts. The remaining files are the interference files for User 1. } \label{Tab:111_2}
		\vspace{-0.5cm}
\end{table}

Consider an ${(1,1,1)}_2-$ICP. In the ${(1,1,1)}_2-$ICP, $\forall k\in [K]$
\begin{itemize}
	\item  want set $\mathcal{W}_k=\{x_{k}\}$ and
	\item  known set \begin{align*}
		\mathcal{K}_k=\{x_{b}: v\in[2], r\in[2],b=<k+v+(v-1)2+r>_8
		\},
	\end{align*}
	i.e.,
	\begin{align*}
		\mathcal{K}_k=\{x_{<k+2>_8}, x_{<k+3>_8}, x_{<k+5>_8}, x_{<k+6>_8}\}.
	\end{align*}
\end{itemize}
\begin{table}[h]
	\begin{minipage}{.44\linewidth}
		\vspace{0.8cm}
		\centering
		\begin{tabular}{| c | }
			\hline 
			\color{red}	$x_{1}$ \\
			\hline 
			$x_{2}$ \\
			\hline 
			\cellcolor{gray}\color{blue}	$x_{3}$ \\
			\hline 
			\cellcolor{gray}\color{blue}	$x_{4}$ \\
			\hline 
			$x_{5}$ \\
			\hline 
			\cellcolor{gray}\color{blue}	$x_{6}$ \\
			\hline 
			\cellcolor{gray}\color{blue}	$x_{7}$ \\
			\hline 
			$x_{8}$ \\
			\hline  
		\end{tabular}
		\vspace{0.2cm}
		\caption{\sl This table corresponds to  the ${(1,1,1)}_2-$ICP. Here, we highlight the requested file of User 1 with the red color fonts and the side-information files of User 1 with the shaded cells and the blue color fonts. The remaining files are the interference files for User 1. } \label{Tab:111}
		\vspace{-0.8cm}
	\end{minipage}%
	\hspace{0.01\linewidth}
	\begin{minipage}{.55\linewidth}
		\centering
		\begin{tabular}{| c | c | c |}
			\hline 
			\color{red}${x_{1,1}}$ &  \color{red}${x_{1,2}}$ & \color{red} ${x_{1,3}}$\\
			\hline 
			$x_{2,1}$ &  $x_{2,2}$ & $x_{2,3}$\\
			\hline
			\cellcolor{gray}\color{blue}$\mathbf{x_{3,1}}$ &  \cellcolor{gray}\color{blue}$\mathbf{x_{3,2}}$ & \cellcolor{gray}\color{blue}$\mathbf{x_{3,3}}$\\
			\hline
			\cellcolor{gray}\color{blue}$x_{4,1}$ &  \cellcolor{gray}\color{blue}$x_{4,2}$ & \cellcolor{gray}\color{blue}$x_{4,3}$\\
			\hline 
			$x_{5,1}$ &  $x_{5,2}$ & $x_{5,3}$\\
			\hline
			\cellcolor{gray}\color{blue}$x_{6,1}$ & \cellcolor{gray}\color{blue} $x_{6,2}$ & \cellcolor{gray}\color{blue}$x_{6,3}$\\
			\hline
			\cellcolor{gray}\color{blue}	$x_{7,1}$ & \cellcolor{gray}\color{blue} $x_{7,2}$ & \cellcolor{gray}\color{blue}$x_{7,3}$\\
			\hline
			$x_{8,1}$ &  $x_{8,2}$ & $x_{8,3}$\\
			\hline
		\end{tabular}
		\vspace{0.2cm}
		\caption{\sl This table corresponds to  the ${(1,1,1)}_2-$ICP with respect to subfiles. Here, we highlight the requested subfiles of User 1 with the red color fonts and the side-information subfiles of User 1 with the shaded cells and the blue color fonts. The remaining subfiles are the interference files for User 1.} \label{Tab:1112bar}
		\vspace{-0.5cm}
	\end{minipage}
\end{table}
Now, split each file into 3 subfiles, i.e., File $x_i$ into 3 equal size subfiles  $x_{i,1}$, $x_{i,2}$ and $x_{i,3}$. Then, with respect to the subfiles, in the ${(1,1,1)}_2-$ICP, $\forall k\in [K]$
\begin{itemize}
	\item  want set $\mathcal{W}_k=\{x_{k,1},x_{k,2},x_{k,3}\}$ and
	\item  known set 
	\begin{align}\label{eqn:111}
		\mathcal{K}_k=\{x_{<k+2>_8,t}, x_{<k+3>_8,t}, x_{<k+5>_8,t}, x_{<k+6>_8,t}:t\in[3]\}.
	\end{align}
\end{itemize}

Observe that \eqref{eqn:111bar} = \eqref{eqn:111}. Therefore,  if we split  the files of the ${(1,1,1)}_2-$ICP into 3 equal size subfiles,  the ${(1,1,1)}_2-$ICP is equal to the $\overline{(1,1,1)}_2-$ICP with respect to the subfiles. We can also observe this using tables \ref{Tab:111_2} and \ref{Tab:1112bar}, where Table \ref{Tab:111_2} corresponds to the $\overline{(1,1,1)}_2-$ICP and Table \ref{Tab:1112bar} corresponds to the ${(1,1,1)}_2-$ICP with subfiles and both are of the same form.

From Theorem \ref{thm:icpub}, an upper bound on the transmission rate of the $\overline{(1,1,1)}_2-$ICP is 8 units.
Since, each subfile is of size 1/3 units,  an upper bound on the transmission rate of the ${(1,1,1)}_2-$ICP is 8/3 units.
\end{proof}

 \subsection{Proof of Corollary \ref{cor:comparison}}\label{sec:corollary7proof}
 \begin{proof}\quad 
 	
 	\textbf{Proof of $R_{\text{new}}(M)\leq R_{\text{HKD}}(M)$:}
  	\begin{align*}
 	R_{\text{new}}(M)&=\frac{\sum_{\mathbf{b }\in \mathcal{B}}\min\{2(K-wL)+w-1-\mathbf{\widehat{b}},K\}}{|\hat{\mathcal{S}}|(w+1)}\\
 	&\leq \frac{K|\mathcal{B}|}{|\hat{\mathcal{S}}|(w+1)}=\frac{K{K-wL+w-1 \choose w}}{\frac{K}{w}{K-wL+w-1 \choose w-1}(w+1)}=\frac{K-wL}{w+1}=R_{\text{HKD}}(M).
 \end{align*} 
 
 \textbf{Proof of $R_{\text{new}}(M)\leq R_{\text{RK}}(M)$:}
 
 In our paper, we use the same placement policy proposed in \cite{reddy2020rate} and hence for a given request pattern, the ICP is the same as \cite{reddy2020rate}.
 	From Lemma \ref{lemma:icpupperbound}, an upper bound on the ICP is given by its local chromatic number, which is defined as the maximum number of different colors that appear in any user's closed anti-outneighborhood, minimized over all proper coloring schemes. The difference between the bound on the server transmission rate derived in \cite{reddy2020rate} and the one proposed in Theorem \ref{thm:multiub2} of this work stems from the choice of the proper coloring scheme. According to our placement policy, for the ICP formed by the multi-access coded caching problem, the number of nodes in the closed anti-outneighborhood of a node is $(K-wL){K-wL+w-1 \choose w}$ \cite{reddy2020rate}. Note that the size of the closed anti-neighborhood only depends on the placement scheme, which is the same in both \cite{reddy2020rate} as well as this work. In \cite{reddy2020rate}, the proper coloring scheme assigns one color per node, i.e., each node in the closed anti-outneighborhood of a node is assigned a different color. Hence, the number of colors in the closed anti-outneighborhood of a node is $(K-wL){K-wL+w-1 \choose w}$ and the data transmission rate is given by $R_{RK}(M)=\frac{(K-wL){K-wL+w-1 \choose w}}{|\mathcal{\widehat{S}}|}=K(1-LM/N)^2$ units. 
 	
 	On the other hand, our proper coloring scheme might repeat the colors across the nodes.  Hence, according to our coloring scheme, the number of colors in the closed anti-outneighborhood of a node is\remove{\color{red}the local chromatic number in our paper is bounded by  some $B$} smaller than or equal to  the number of nodes in the closed anti-outneighborhood of a node, which is $(K-wL){K-wL+w-1 \choose w}$. Therefore, the data transmission rate $R_{new}(M)$ in our paper is upper bounded by $\frac{(K-wL){K-wL+w-1 \choose w}}{|\mathcal{\widehat{S}}|}=K(1-LM/N)^2=R_{RK}(M)$ units.
 \remove{Note that in this paper, we use the same placement policy proposed in \cite{reddy2020rate} and hence for a given request pattern, the ICP is the same as \cite{reddy2020rate}. The difference between the bound on the server transmission rate derived in \cite{reddy2020rate} and the one proposed in Theorem \ref{thm:multiub2} of this work stems from the choice of the proper coloring scheme and the corresponding bound on the local chromatic number.  In \cite{reddy2020rate}, the proper coloring scheme assigns one color per node, i.e., all the nodes in the closed anti-outneighborhood of a node contain different colors. On the other hand, our proper coloring scheme might repeat the colors across the nodes and thus the bound proposed in Theorem \ref{thm:multiub2} is no larger than the one in \cite{reddy2020rate}. 
 }
 \end{proof}

\end{document}